\pdfoutput=1
\documentclass[journal,onecolumn]{IEEEtran}
\usepackage{graphicx}
\usepackage{cite}
\usepackage{booktabs}
\usepackage{comment}
\usepackage{epstopdf}
\usepackage{amsthm}
\usepackage{amssymb}
\usepackage{amsmath}
\usepackage{mathtools}

\usepackage{xcolor}
\definecolor{linkblue}{RGB}{0,70,160}
\definecolor{citationpurple}{RGB}{128,0,128}
\usepackage[colorlinks=true,linkcolor=linkblue,citecolor=citationpurple,urlcolor=linkblue,bookmarks=false,hypertexnames=false]{hyperref}
\newtheorem{lemma}{Lemma}

\newtheorem{corollary}{Corollary}
\newtheorem{proposition}{Proposition}
\newtheorem{theorem}{Theorem}
\newcounter{algorithm}
\newcounter{algline}
\newenvironment{compactalgorithm}[1]{%
  \refstepcounter{algorithm}%
  \setcounter{algline}{0}%
  \par\begingroup\small
  \setlength{\tabcolsep}{0pt}%
  \setlength{\parskip}{0pt}%
  \setlength{\parsep}{0pt}%
  \noindent\begin{minipage}{0.96\textwidth}
  \hrule\vspace{2pt}
  \textbf{Algorithm~\thealgorithm}\quad #1\par
  \vspace{2pt}\hrule\vspace{3pt}
  \begin{tabular}{@{}r@{\hspace{0.75em}}p{0.90\textwidth}@{}}
}{%
  \end{tabular}
  \vspace{2pt}\hrule
  \end{minipage}\par
  \endgroup
}
\newcommand{\alginput}[1]{\multicolumn{2}{@{}p{0.96\textwidth}@{}}{\hspace{2.7em}\textbf{Input:} #1}\\[-1pt]}
\newcommand{\algoutput}[1]{\multicolumn{2}{@{}p{0.96\textwidth}@{}}{\hspace{2.7em}\textbf{Output:} #1}\\}
\newcommand{\algline}[1]{\stepcounter{algline}\textbf{\thealgline} & #1\\[-1pt]}

\begin{document}
\title{Age of Information in Time-Varying Multi-Priority Queues}
\author{Burak Karasakal, Aimin Li, \IEEEmembership{Member, IEEE}, and Elif Uysal, \IEEEmembership{Fellow, IEEE}%
\thanks{The authors are with the \href{https://cng-eee.metu.edu.tr/}{Communication Networks Research Group (CNG)}, METU, Ankara, T\"urkiye (e-mail: \{burak.karasakal, aimin, uelif\}@metu.edu.tr).}
\thanks{This work is supported by the European Union through ERC Advanced Grant 101122990-GO SPACE-ERC-2023-A. Views and opinions expressed are however those of the author(s) only and do not necessarily reflect those of the European Union or the European Research Council Executive Agency. Neither the European Union nor the granting authority can be held responsible for them.}}

\maketitle

\begin{abstract}
In networks with intermittent connectivity, such as mobile, aerial, and space systems, maintaining information freshness is complicated by \textit{time-varying arrivals}, \textit{service disruptions}, and interactions among traffic classes with different \textit{priorities}. To capture these effects, we study a multi-priority single-server queue with \textit{time-varying} arrivals and service rates under intermittent connectivity. Our main result shows that an appropriately selected collection of state-conditioned first moments closes exactly, leading to a finite-dimensional linear time-periodic Ordinary Differential Equation (ODE) system for the mean Age of Information (AoI) and mean Peak Age of Information (PAoI) of each priority class. For periodic arrival and service rates, we define a one-period state map by propagating the ODE over a single period, and use the periodicity condition to formulate the periodic steady state as a fixed point of this map. We then propose a fixed-point iteration algorithm and prove its convergence to the unique periodic steady state (PSS). Numerical results reveal that high-priority traffic can strongly reshape the service process seen by lower-priority classes.
\end{abstract}


\section{Introduction}
\subsection{Background and Motivation}
   Timely status delivery is essential in mobile and intermittently connected cyberphysical systems, including autonomous driving, satellite sensing, and UAV relaying, where goal-oriented control and coordination decisions depend on recently generated information~\cite{yates2021survey,kadota2018scheduling}.
   In such systems, the value of an update is determined not only by whether it is eventually delivered, but also by whether it still describes the current state of the physical process, namely the freshness of information. 
   Age of Information (AoI), defined as the time elapsed since the generation time of the freshest update available at the monitor~\cite{kaul2012real,yates2021survey}, captures this receiver-side notion of freshness.

Freshness analysis becomes particularly challenging when mobility, intermittent connectivity, and heterogeneous traffic coexist. 
First, update generation may \textit{vary over time} due to energy-aware sampling, event-triggered sensing, or mission-dependent observation schedules. 
Second, service capability may also be \textit{time-varying}, as transmission opportunities disappear during outages and effective service rates fluctuate with contact geometry, channel quality, shared-spectrum contention, or bandwidth allocation. Third, status streams are rarely homogeneous: safety-critical control packets, alarms, sensing reports, and routine telemetry may have substantially different urgency levels and freshness requirements. 
This makes \textit{priority differentiation} operationally necessary rather than merely analytically convenient~\cite{yates2019multiple,xu2021peakpriority}. 
Finally, under \textit{intermittent connectivity}, multiple updates from the same class may be generated before the next service opportunity; older waiting updates are then often dominated by fresher ones, motivating latest-only packet management~\cite{costa2016packetmgmt,najm2017harq}.

These considerations give rise to time-varying multi-priority status-update queues. In such systems, the freshness of a low-priority class is determined not merely by the nominal service rate, but by the residual service capacity left after higher-priority traffic is served.
This residual-service dependence induces \textit{priority coupling}: the effective service process experienced by each class becomes endogenous to the joint system state.
When the arrivals, service rates, and connectivity patterns \textit{vary periodically}, as in satellite orbital passes, duty-cycled communication, or scheduled resource allocation, \textit{the relevant long-run object is not a constant steady state but a periodic steady state.}


    \subsection{Related Work}

With the long-run behavior governed by a \textit{periodic} rather than \textit{time-invariant} steady state, it is natural to look beyond stationary mean-age metrics. A complementary line of work has studied distributional and non-stationary aspects of AoI, including stationary AoI distributions~\cite{inoue2019stationarydist}, violation probabilities in latest-only buffer systems~\cite{champati2019distribution}, and SHS-based moment characterizations for Markovian networks~\cite{yates2020moments}. 
Related distributional analyses have considered preemptive multi-source queues~\cite{dogan2021probpreemptive,moltafet2022mgf,hu2021violation}, while recent work has developed multi-dimensional PDE frameworks for time-varying $M_t/G/1/1$ systems~\cite{xu2025distribution}. 
Classical non-stationary queueing approximations also provide useful intuition for time-varying systems~\cite{green2007coping,feldman2008staffing}. 
However, these approaches do not provide an explicit class-wise AoI and PAoI characterization for time-varying multi-priority queues with latest-only replacement. 
The missing difficulty is the priority-induced residual service: a lower-priority class does not observe the physical service process directly, but only the service opportunities left after higher-priority traffic has been served.

This coupling makes the problem fundamentally different from applying stationary AoI formulas with time-dependent rates. 
Time variation affects the external arrival and service primitives, while priority decisions endogenously reshape the effective service opportunities available to lower-priority classes. 
At the same time, latest-only replacement makes the age of the next delivered packet depend on the evolution of the per-class buffers, not only on the current server state. 
Consequently, class-wise AoI and PAoI cannot be characterized by separate single-class balance equations. 
A tractable analysis requires a joint description of the server state, replacement buffers, packet ages in service and in waiting, and class-specific completion events. 
This motivates an exact finite-dimensional first-moment framework for computing the periodic freshness trajectories of all priority classes.

\subsection{Contributions}

The main contributions of this paper are:
\begin{enumerate}
    \item \textbf{Time-varying priority freshness model.}
    We formulate a multi-class status-update queue with strict priority service, latest-only per-class replacement, time-varying arrivals and service rates, and intermittent ON/OFF service availability. 
    The model captures the residual-service-opportunity effect through which higher-priority traffic endogenously reshapes the service opportunities observed by lower-priority classes.

    \item \textbf{Exact closed ODE.}
     We show that an appropriately selected collection of state-conditioned first moments closes exactly, despite three coupled difficulties: priority-induced endogenous residual service, latest-only replacement across per-class buffers, and a periodic rather than stationary long-run regime. This yields a finite-dimensional linear time-periodic ODE system that jointly characterizes the class-wise mean AoI and completion-conditioned mean PAoI trajectories.

    \item \textbf{Periodic steady-state computation.}
    For periodic arrivals, service rates, and availability patterns, we define the one-period state map induced by the ODE system and formulate the periodic steady state as its \textit{fixed point}. 
    We establish existence and uniqueness of the periodic solution and provide a relaxed \textit{fixed-point iteration} for computing the resulting freshness trajectories.

    \item \textbf{Structural insights on low-priority staleness.}
    The framework reveals how residual service opportunities govern low-priority freshness under intermittent connectivity. 
    It further identifies regimes in which time-average AoI and completion-conditioned mean PAoI separate sharply; in particular, long unserved intervals can cause the mean AoI of low-priority traffic to exceed the PAoI observed at completion epochs.
\end{enumerate}

\section{Model and Formulation}
\subsection{System Model}

We consider a multi-priority queueing system designed to process time-sensitive information packets from $N$ independent sources (classes) to a destination, as shown in Fig.~\ref{fig:sys_model}. The system consists of $N$ classes with distinct priorities, a size-one waiting buffer for each class, and a single server. We characterize the system by the class currently in service and the occupancy of the waiting buffers. Specifically, we define the service state $J \in \{0,1,\ldots,N\}$ as the index of the class currently being served, where $J=0$ indicates that the server is idle, and the buffer state $B_i \in \{0,1\}$ for class $i \in \{1,\ldots,N\}$, where $B_i=1$ indicates that a packet is waiting and $B_i=0$ indicates that the buffer is empty. The system state is therefore represented by the $(N+1)$-tuple $(J,B_1,\ldots,B_N) \in \mathcal{Q}$, where $\mathcal{Q}$ denotes the finite state space. For any state $s \in \mathcal{Q}$, $J(s)$ and $B_i(s)$ denote the corresponding components of $s$. We now summarize the key properties of the system.

\begin{itemize}
    \item \textbf{Priority Discipline:} The system operates under a strict non-preemptive priority discipline. The scheduling rule is characterized by the map $\mathrm{next}: \mathcal{Q} \rightarrow \{0,1,\ldots,N\}$, which selects the next class after a service completion:
    \begin{align}
        \mathrm{next}(s) = \begin{cases} \min\mathcal{A}, & \mathcal{A} \neq \emptyset, \\ 0, & \text{otherwise,} \end{cases}
    \end{align}
    where $\mathcal{A} = \{i \in \{1,\ldots,N\} : B_i(s) = 1\}$ is the set of classes with a packet waiting in their buffer at state $s$. Thus, a lower-priority packet is served only when all higher-priority waiting buffers are empty.
    \item \textbf{Buffer Management:} Each class maintains a size-one waiting buffer operating under a \textit{preemptive replacement (latest-only)} policy. The buffer is emptied immediately when a packet enters service, at which point the packet becomes \textit{non-preemptible} until service completion. This policy ensures that the most recently generated packet is served next, thereby improving the freshness of the information delivered to the destination.
    \item \textbf{Arrival \& Service Processes:} Arrivals for class $i$ follow independent Poisson processes with time-varying rates $\lambda_i(t)$. Service times follow an exponential distribution with time-varying nominal rates $\mu_i(t)$.
    \item \textbf{Intermittent Connectivity:} The link between the server and the destination is subject to intermittent outages, causing the service rate $\mu_i(t)$ for class $i$ to vary over time and to drop to zero during outage periods.
    \item \textbf{Inherent Periodicity:} To account for physical cycles such as satellite orbital motion and daily network traffic patterns, all rates and availability functions are periodic with a common period $T$ \footnote{This assumption is not overly restrictive when the individual periods are rationally related; in that case, a common period can be formed by their least
common multiple.} Formally, $\lambda_i(t) = \lambda_i(t+T)$ and $\mu_i(t) = \mu_i(t+T)$ for all $i$.
\end{itemize}
    \begin{figure}[!t] 
    \centering
    \includegraphics[width=0.62\textwidth]{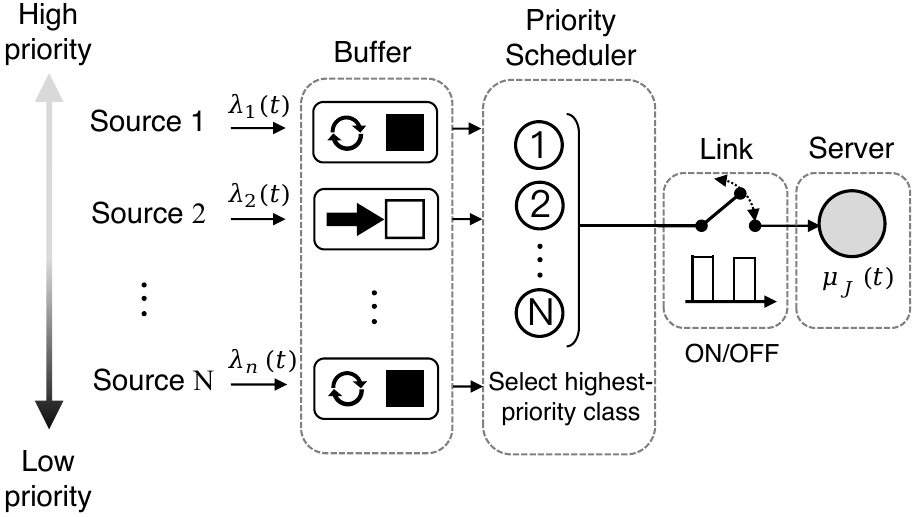}
    \caption{System model for a multi-priority latest-only status-update queue with time-varying arrival rate, a strict non-preemptive scheduler, time-varying service rate, and intermittent
link availability.}
    \label{fig:sys_model}
    \end{figure}
\subsection{Age of Information}

The Age of Information (AoI) for class $i$ is defined as $\Delta_i(t) \triangleq t - \max\{r_n^{(i)} : C_n^{(i)} \le t\}$, where $r_n^{(i)}$ and $C_n^{(i)}$ denote the generation time and completion time, respectively, of the $n$-th delivered update. We then introduce the following mean evolutions:
\begin{itemize}
    \item \textbf{Mean AoI:} $\bar{\Delta}_i(t) \triangleq \mathbb{E}[\Delta_i(t)]$.
    \item \textbf{Mean PAoI:} $\hat{\Delta}_i(t) \triangleq \text{lim}_{h\rightarrow0^{+}}\mathbb{E}[\Delta_i^{-}(t) \mid \mathcal{E}_i(t,t+h)]$, defined for all $t$ such that $P\{\mathcal{E}_i(t, t+h)\} > 0$ for $\forall h > 0$, where $\mathcal{E}_i(t,t+h)$ denotes the event that a class‑$i$ completion occurs in the infinitesimal interval $(t,t+h)$ and $\Delta_i^{-}(t) \triangleq \text{lim}_{k \rightarrow 0^{-}}\Delta_i(t+k)$. 
    \item \textbf{Unserved Mean Age:} $\Delta_i^c(t) \triangleq \mathbb{E}[\Delta_i(t) \mid J(t) \neq i]$, where $J(t)$ denotes the class in service at time $t$.
\end{itemize}

\subsection{Continuous-Time Markov Chain Formulation}
We formulate a Continuous-Time Markov Chain (CTMC) to capture the discrete transitions between the different queue and server states. We index the states in $\mathcal{Q}$ through the mapping $\sigma: \mathcal{Q} \rightarrow \{1, \ldots, |\mathcal{Q}|\}$ defined as:
\begin{align}
    \sigma(s) = \begin{cases} 1, & J(s) = 0, \\ 2 + 2^N(J(s)-1) + \displaystyle\sum_{i=1}^N 2^{i-1}B_i(s), & J(s) \neq 0. \end{cases}
\end{align}
The system state at time $t$ is described by a discrete process $q(t) \in \mathcal{Q}$. Let $p_{\sigma(s)}(t) = \Pr\{q(t) = s\}$ denote the probability that the system occupies state $s \in \mathcal{Q}$ at time $t$. We define the state probability row vector as $\mathbf{p}(t) = [p_{1}(t), \ldots, p_{|\mathcal{Q}|}(t)] \in [0, 1]^{1 \times |\mathcal{Q}|}$. Its evolution is governed by the Kolmogorov Forward Equation: $\dot{\mathbf{p}}(t) = \mathbf{p}(t)\mathbf{Q}(t)$,
where $\mathbf{Q}(t) \in \mathbb{R}^{|\mathcal{Q}| \times |\mathcal{Q}|}$ is the generator matrix. Let $a^{(k,+)}$ denote the post-arrival state when a class-$k$ arrival changes the CTMC state: an idle server starts serving class $k$, while a busy server fills the class-$k$ buffer if $B_k(a)=0$. Replacement arrivals are self-transitions. Let $a^{-}$ denote the post-completion state, with the server moved to $\mathrm{next}(a)$ and the selected buffer cleared. With $Q_{i,i}(t)=-\sum_{m\neq i}Q_{i,m}(t)$, the off-diagonal rates are, for $a\neq b$,
\begin{align}
Q_{\sigma(a),\sigma(b)}(t)=
\begin{cases}
\lambda_k(t), & b=a^{(k,+)},\ k\in\{1,\ldots,N\},\\
\mu_{J(a)}(t), & b=a^{-},\ J(a)\neq0,\\
0, & \text{otherwise,}
\end{cases}
\label{eq:generator_entries}
\end{align}

\subsection{Auxiliary Age Variables}

Direct analysis of $\bar{\Delta}_i(t)$ and $\hat{\Delta}_i(t)$ is intractable without auxiliary variables, because the age at delivery depends not only on the current system state but also on how long the packet in service and the waiting packets have been held. We therefore introduce continuous auxiliary variables. Let $Y(t)$ denote the age of the packet in service, and let $Z_i(t)$ denote the current waiting time of the packet in the class-$i$ buffer. Each variable is zero when the corresponding server or buffer is idle. We then define the state-dependent moment row vectors $\mathbf{a}_i(t)$, $\mathbf{y}(t)$, $\mathbf{z}_i(t) \in \mathbb{R}^{1 \times |\mathcal{Q}|}_{\geq 0}$ with entries:
\begin{align}
    a_{i,\sigma(s)}(t) &\triangleq \mathbb{E}[\Delta_i(t) 1\{q(t)=s\}], \\
    y_{\sigma(s)}(t) &\triangleq \mathbb{E}[Y(t) 1\{q(t)=s\}], \\
    z_{i,\sigma(s)}(t) &\triangleq \mathbb{E}[Z_i(t) 1\{q(t)=s\}],
\end{align}
where $1\{\cdot\}$ is the indicator function. Let $\mathrm{dest}(s)$ denote the state reached immediately after a service completion. We define the following sparse transition rate matrices:
\begin{align}
    \mathbf{M}_{\sigma(s), \sigma(\text{dest}(s))}^{(\text{comp})}(t) &= \mu_{J(s)}(t) \cdot 1\{J(s) \neq 0\}, \\
    \mathbf{M}_{\sigma(s), \sigma(\text{dest}(s))}^{(i)}(t) &= \mu_i(t) \cdot 1\{J(s) = i\}, \\
    \mathbf{M}_{\sigma(s), \sigma(\text{dest}(s))}^{(\text{next}=i)}(t) &= \mu_{J(s)}(t) \cdot 1\{\text{next}(s) = i\}.
\end{align}
We set to zero all entries that do not correspond to service departures. The resulting matrices enable a compact ODE characterization of the age dynamics, which we derive in the next section.

\section{Main Results}
\subsection{Governing Equations for System Evolution}
We next derive governing equations that characterize the joint evolution of the queueing state and the associated age moments. The main technical difficulty is that priority coupling and latest-only replacement render the effective service seen by each class state-dependent and endogenous, preventing standard decoupled or stationary AoI analyses. The following theorem establishes a closed, finite-dimensional system of linear time-periodic ODEs for the relevant moment vectors. This characterization enables direct numerical evaluation of class-wise AoI and PAoI, and forms the basis for studying convergence to the periodic steady state. Throughout, for any state predicate $\mathcal{C}(s)$ on $s\in\mathcal{Q}$, we define the row indicator vector $\boldsymbol{\delta}_{\mathcal{C}}\in\{0,1\}^{1\times|\mathcal{Q}|}$ by $[\boldsymbol{\delta}_{\mathcal{C}}]_{\sigma(s)}=\mathbf{1}\{\mathcal{C}(s)\}$ for all $s\in\mathcal{Q}$.
\begin{theorem}[ODEs for Age Moments and Recovery of Performance Metrics] \label{theorem_main}
The state-dependent moment vectors $\mathbf{a}_i(t)$, $\mathbf{y}(t)$, and $\mathbf{z}_i(t)$ satisfy the following closed system of linear time-periodic ODEs for each class $i \in \{1, \ldots, N\}$:

\textbf{(i) AoI moments:}
\begin{equation}
    \dot{\mathbf{a}}_i(t) = \mathbf{p}(t) + \mathbf{a}_i(t)\mathbf{Q}(t) + \bigl(\mathbf{y}(t) - \mathbf{a}_i(t)\bigr)\mathbf{M}^{(i)}(t),
\end{equation}

\textbf{(ii) In-service packet age moments:}
\begin{equation}
    \begin{split}
        \dot{\mathbf{y}}(t) = \,&\mathbf{p}(t) \odot \boldsymbol{\delta}_{J\neq 0} + \mathbf{y}(t)\mathbf{Q}(t) \\
        &+ \bigl(\mathbf{z}^{\mathrm{next}}(t) - \mathbf{y}(t)\bigr) \mathbf{M}^{(\mathrm{comp})}(t),
    \end{split}
\end{equation}
where $\mathbf{z}^{\mathrm{next}}(t) = \sum_{k=1}^{N} \mathbf{z}_k(t) \odot \boldsymbol{\delta}_{\mathrm{next}=k}$.

\textbf{(iii) Waiting packet age moments:}
\begin{align}
    \dot{\mathbf{z}}_i(t) &= \mathbf{p}(t)\odot \boldsymbol{\delta}_{B_i = 1} + \mathbf{z}_i(t)\mathbf{Q}(t) - \lambda_i(t)\big( \mathbf{z}_i(t)\odot \boldsymbol{\delta}_{B_i = 1} \big) \notag \\ 
    &- \mathbf{z}_i(t)\mathbf{M}^{(\mathrm{next}=i)}(t).
\end{align}

Furthermore, the average AoI $\bar{\Delta}_i(t)$ and average peak AoI $\hat{\Delta}_i(t)$ for each class $i$ are recovered from the state-dependent moment vector $\mathbf{a}_i(t)$ and the state distribution vector $\mathbf{p}(t)$ via:
\begin{equation}
    \bar{\Delta}_i(t) = \mathbf{a}_i(t)\mathbf{1}, \quad \hat{\Delta}_i(t) = \frac{\mathbf{a}_i(t)\mathbf{1}_{J=i}}{\mathbf{p}(t)\mathbf{1}_{J=i}}.
\end{equation}
\end{theorem}

\begin{proof}
See Appendix~\ref{appendix_a}.
\end{proof}

The ODE terms in Theorem~\ref{theorem_main} have direct physical interpretations. In (10), $\mathbf{p}(t)$ captures the unit-rate growth of monitor age, $\mathbf{a}_i(t)\mathbf{Q}(t)$ transports accumulated age mass across CTMC state transitions, and $\bigl(\mathbf{y}(t)-\mathbf{a}_i(t)\bigr)\mathbf{M}^{(i)}(t)$ resets the class-$i$ monitor age to the in-service packet age upon class-$i$ completions. Similarly, the first terms in (11) and (12) represent unit-rate aging of packets currently in service or waiting, while the remaining matrix terms transfer or remove packet-age mass when completions, priority selections, and latest-only replacements occur.

\subsection{Asymptotic Convergence to the Periodic Steady State} 


\begin{theorem}[Existence and Uniqueness of the Periodic Steady State]
\label{existence_thm_main}
The system of linear ODEs established in Theorem~\ref{theorem_main} admits a unique periodic steady state, denoted by $\mathbf{x}^\star(t)$, satisfying $\mathbf{x}^\star(t+T) = \mathbf{x}^\star(t)$ for all $t \geq 0$. Moreover, for any initial condition $\mathbf{x}(0)$, the solution converges exponentially to $\mathbf{x}^\star(t)$ in the standard Euclidean norm, satisfying:
\begin{equation}
    \|\mathbf{x}(t) - \mathbf{x}^\star(t)\|_2 \leq m e^{-\gamma t} \|\mathbf{x}(0) - \mathbf{x}^\star(0)\|_2,
\end{equation}
for some constants $\gamma > 0$ and $m \geq 1$.
\end{theorem}

\begin{proof}
See Appendix~\ref{appendix_b}.
\end{proof}

\subsection{The Gap Between Mean and Peak Age}

\begin{proposition}[Gap Between Mean and Peak Age]\label{gap_prop}
The gap between the average peak age and the average age of class $i$ at time $t$ is given by 
\begin{equation}
    \hat{\Delta}_i(t) - \bar{\Delta}_i(t) = (1 - \pi_i(t))(\hat{\Delta}_i(t) - \Delta^{c}_i(t)),
    \label{eq:gap}
\end{equation}
where $\pi_i(t) \triangleq \mathbf{p}(t)\boldsymbol{\delta}_{J=i}$ denotes the probability that class $i$ is being served at time $t$.
\end{proposition}

\begin{proof}
See Appendix~\ref{appendix_c}.
\end{proof}

Proposition~\ref{gap_prop} shows that the sign of the gap is not intrinsic to ``peak'' versus ``average'' age. Since $1-\pi_i(t)\ge 0$, the difference $\hat{\Delta}_i(t)-\bar{\Delta}_i(t)$ is entirely determined by $\hat{\Delta}_i(t)-\Delta_i^{c}(t)$. Thus, despite common intuition, the (average) peak age is not necessarily an upper bound on the (average) age; the ordering depends on how often the class is in service and on the conditional age experienced in those service states.

\section{Numerical Solution}\label{mumericalsolution}

This section presents a numerical algorithm for computing the periodic steady state (PSS) of the system of linear ODEs established in Theorem \ref{theorem_main}, whose existence and uniqueness are guaranteed by Theorem~\ref{existence_thm_main}. Let $F: \mathbb{R}^{d} \to \mathbb{R}^{d}$ denote the one-period map that takes an initial condition $\mathbf{x}(0)$ to the solution $\mathbf{x}(T)$ obtained by integrating the system of ODEs in Theorem~\ref{theorem_main} over one period $T$, where
\begin{equation}
 \mathbf{x}(t) = [\mathbf{a}_1(t), \dots, \mathbf{a}_N(t), \mathbf{y}(t), \mathbf{z}_1(t), \dots, \mathbf{z}_N(t), \mathbf{p}(t)]^T \in \mathbb{R}^{d}
\end{equation}
  is the stacked state vector of dimension $d = (2N+2)|\mathcal{Q}|$. A PSS initial condition 
$x^{\star}(0)$ is a fixed point of $F$, i.e.,
\begin{equation}
    x^{\star}(0) = F\bigl(x^{\star}(0)\bigr),
\end{equation}
ensuring $x^{\star}(t + T) = x^{\star}(t)$ for all $t \geq 0$. To solve for $x^{\star}(0)$, 
we use the relaxed fixed-point iteration summarized in Appendix~\ref{app:pss_algorithm}. Given a tolerance $\varepsilon$, maximum iterations $K$, and a relaxation parameter $\alpha \in (0,1]$ used to dampen numerical oscillations, the iteration updates the initial 
condition until the relative residual satisfies the stopping criterion $\varepsilon$ or the iteration limit $K$ is reached. Furthermore, as numerical solvers may introduce small floating-point drift, we apply a renormalization step after each period integration. This involves rescaling the probability block to enforce the unit-sum constraint:
\begin{equation}
    p_s \leftarrow \frac{p_s}{\sum_{s'} p_{s'}}, \qquad \forall\, s,
    \label{eq:renorm}
\end{equation}
ensuring consistency across the state vector blocks before the relaxed update is performed. The convergence of the procedure follows directly from the mathematical construction used in the proof of Theorem~\ref{existence_thm_main}.

\begin{figure}[!t]
    \centering
    \includegraphics[width=0.82\textwidth,trim=7bp 6bp 8bp 7bp,clip]{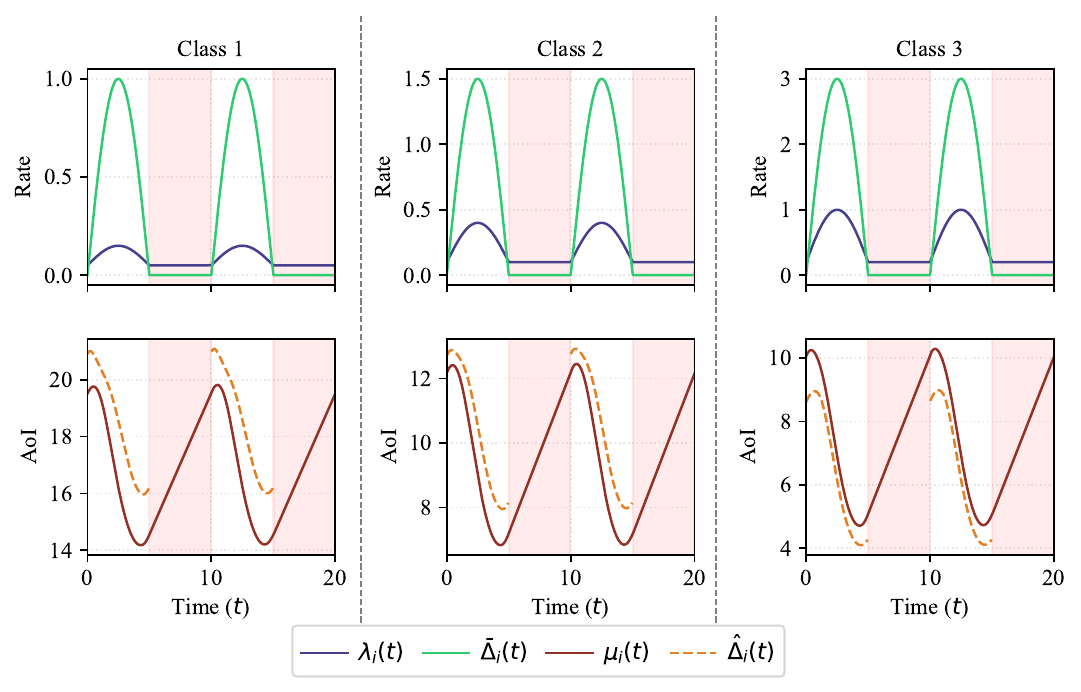}
    \caption{Rates and AoI metrics for three distinct priority classes. The top row illustrates the arrival rates ($\lambda_i$) and service rates ($\mu_i$), while the bottom row depicts the corresponding evolution of the mean AoI ($\bar{\Delta}_i$) and peak AoI ($\hat{\Delta}_i$). Red-shaded regions indicate periods of link outages.}
    \label{fig:rates_aoi_dynamics}
    \label{fig:aoi_dynamics}
    \end{figure}

\begin{corollary} \label{fixed_point}
The fixed-point iteration described in Appendix~\ref{app:pss_algorithm} converges to the unique periodic steady state $\mathbf{x}^\star(t)$ satisfying $\mathbf{x}^\star(t+T) = \mathbf{x}^\star(t)$ for all $t \geq 0$.
\end{corollary}

\begin{proof}
See Appendix~\ref{appendix_d}.
\end{proof}

    \section{Simulation Results}


        \subsection{Parameter Settings}

        We consider a time-varying multi-priority queueing system with a cycle period of $T = 10.0$ and an active service window of $T_{\text{pass}} = 5.0$, corresponding to a $50\%$ duty cycle. The system serves $N=3$ priority classes. The parameter set is chosen to stress priority coupling: lower-priority classes have larger arrival bursts, so their freshness is governed by the residual capacity left after higher-priority service. The instantaneous service and arrival rates are modeled using a windowed sinusoid. For each class $i$, the service rate is defined as:
    \begin{equation}
    \mu_i(t) = \mu_{i}^{\text{peak}} \max\left(0, \cos\left(\frac{\pi}{T_{\text{pass}}} \left(t - \frac{T_{\text{pass}}}{2}\right)\right)\right),
    \end{equation}
    when the link is available, and zero otherwise. The arrival process for each class follows a similar profile with a constant baseline:
    \begin{equation}
    \lambda_i(t) = \lambda_{i}^{\text{base}} + \lambda_{i}^{\text{peak}} \max\left(0, \cos\left(\frac{\pi}{T_{\text{pass}}} \left(t - \frac{T_{\text{pass}}}{2}\right)\right)\right),
    \end{equation}

    The baseline and peak rate parameters for each priority class are provided in Table \ref{tab:parameters}. To validate the analytical ODE framework, we compare the periodic steady-state metrics with progressive Monte Carlo simulations using up to $10{,}000$ sample paths averaged over $100$ trials. Fig.~\ref{fig:validation} shows that both mean-AoI and peak-AoI MAEs decrease rapidly as the number of instances increases, confirming that the ODE captures both continuous age growth and completion-time resets. Fig.~\ref{fig:rates_aoi_dynamics} shows the resulting long-term age trajectories, while Fig.~\ref{fig:probs} reports the corresponding idle and service-state probabilities.

    \begin{table}[!t]
    \centering
    \caption{Simulation Parameters for Priority Classes}
    \label{tab:parameters}
    \begin{tabular}{@{}lcccc@{}}
    \toprule
    \textbf{Parameter} & \textbf{Symbol} & \textbf{Class 1} & \textbf{Class 2} & \textbf{Class 3} \\ \midrule
    Peak Service Rate  & $\mu^{\text{peak}}$ & 1.0 & 1.5 & 3.0 \\
    Base Arrival Rate  & $\lambda^{\text{base}}$ & 0.05 & 0.10 & 0.20 \\
    Peak Arrival Rate  & $\lambda^{\text{peak}}$ & 0.10 & 0.30 & 0.80 \\ \bottomrule
    \end{tabular}
    \end{table}

    \begin{figure}[!t]
    \centering
    \begin{minipage}[t]{0.48\textwidth}
    \centering
    \includegraphics[width=\linewidth,trim=7bp 6bp 8bp 6bp,clip]{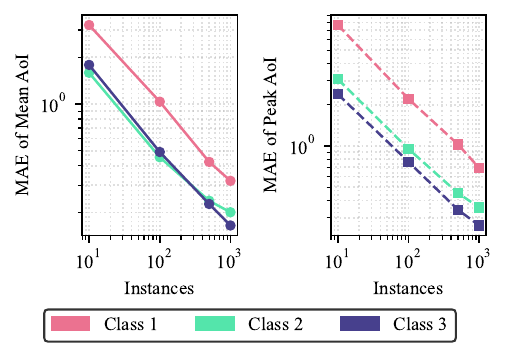}
    \caption{Validation results. The mean-AoI and peak-AoI mean absolute errors (MAEs) are shown as functions of the number of Monte Carlo instances.}
    \label{fig:validation}
    \end{minipage}\hfill
    \begin{minipage}[t]{0.48\textwidth}
    \centering
    \includegraphics[width=0.74\linewidth,trim=7bp 6bp 8bp 6bp,clip]{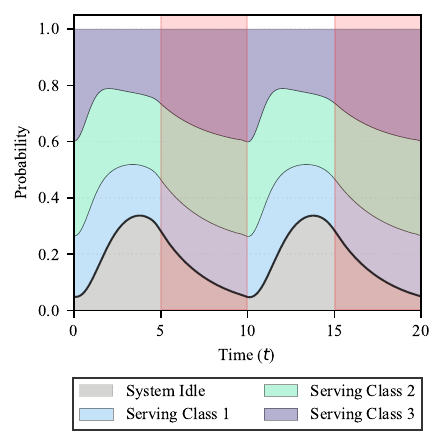}
    \caption{The evolution of the system idle and service state probabilities. Red-shaded regions indicate link outages.}
    \label{fig:probs}
    \end{minipage}
    \end{figure}


 	
    \subsection{Discussion}
    The trajectories reveal two related effects. First, higher-priority classes track the active service windows more closely, whereas lower-priority packets can wait through outages and high-priority busy periods. Second, the average Peak Age of Information (PAoI) may fall below the time-average AoI. Average AoI captures continuous aging during starvation and outages, whereas conditional PAoI is sampled only at completion epochs and can miss extreme age growth. As shown in Fig.~\ref{fig:rates_aoi_dynamics}, the lowest-priority class exhibits this inversion under intermittent connectivity, highlighting the need to evaluate time-varying freshness trajectories rather than relying only on aggregate metrics.



    Proposition~\ref{gap_prop} formalizes this separation through the service probability $\pi_i(t)$. Fig.~\ref{fig:probs} shows that the server spends substantial time either idle or serving higher-priority classes, which lowers the service probability of the lowest-priority class during parts of the cycle. Since $\pi_i(t)<1$, the factor $(1-\pi_i(t))$ in \eqref{eq:gap} is nonzero; the two age metrics differ whenever the conditional age during service differs from that outside service.




\section{Conclusion}

This paper develops an exact ODE-based framework for AoI in multi-priority networks with intermittent connectivity. The resulting fixed-point method computes the unique periodic solution and reveals how residual service can make mean AoI exceed mean PAoI for low-priority traffic. This suggests that priority design in intermittent systems should account for the full time-varying freshness trajectory, not only completion-sampled metrics.


\appendices
\allowdisplaybreaks[4]
\setlength{\jot}{2pt}

\section{Proof of Theorem \ref{theorem_main}}\label{appendix_a}

The proof comprises five subsections. Subsections 1–3 derive ODEs for the AoI moments, the in-service packet age moments, and the waiting packet age moments, respectively. Subsections 4–5 then establish identities for the mean AoI and peak AoI in terms of the AoI moments.

\subsection{AoI Moments}

Let $\mathcal{E}(t,t+h)$ denote the event that a completion occurs in $(t,t+h)$.  We partition the state space based on transition likelihood under conditioning on $\mathcal{E}(t,t+h)$. We define the set of states from which a completion is possible in the sense that:
\begin{align}
    \mathcal{A}_s &= \bigl\{ m \in \mathcal{Q} : \lim_{h \to 0} \frac{P\{q(t)=m,\, q(t+h)=s \mid \mathcal{E}(t,t+h)\}}{h} > 0 \bigr\},
\end{align}
for a given $s \in \mathcal{Q}$ i.e., $\mathcal{A}_s$ contains the states from which a completion can trigger a transition to state $s$ with a positive rate. That is, the conditional probability of a transition occurring to state $s$ given the event $\mathcal{E}(t,t+h)$ is proportional to $h$. We expand $a_{i,s}(t+h)$ by inserting indicator functions for all states at time $t$. Partitioning the state space into $\mathcal{A}_s$ and $\mathcal{A}_s^c$ lets us handle separately the cases where a service completion does and does not occur.
\begin{subequations}
\refstepcounter{equation}
\begin{align} \label{iterated}
a_{i,\sigma(s)}&(t+h) \overset{\mathclap{(\alph{equation})}}{=} \mathbb{E}\big(\Delta_i(t+h)\, 1\{q(t+h) = s\}\big) \tag{\theparentequation} \\ \noalign{\refstepcounter{equation}}
&\overset{\mathclap{(\alph{equation})}}{=} \mathbb{E}\big(\Delta_i(t+h)\, 1\{q(t+h) = s\}\, \sum_{m} 1\{q(t) = m\}\big) \notag  \\ \noalign{\refstepcounter{equation}}
&\overset{\mathclap{(\alph{equation})}}{=} \sum_{m \in \mathcal{A}_s} \mathbb{E}\big(\Delta_i(t+h)\, 1\{q(t) = m,\ q(t+h) = s\}\big) + \sum_{m \in \mathcal{A}_s^c} \mathbb{E}\big(\Delta_i(t+h)\, 1\{q(t) = m,\ q(t+h) = s\}\big) \notag \\ \noalign{\refstepcounter{equation}}
&\overset{\mathclap{(\alph{equation})}}{=} \sum_{m \in \mathcal{A}_s} \mathbb{E}(Y(t+h) 1\{q(t)=m, q(t+h)=s\}) + \sum_{m \in \mathcal{A}_s^c} \mathbb{E}(\Delta_i(t+h) 1\{q(t)=m, q(t+h)=s\}) + o(h)
\notag \\ \noalign{\refstepcounter{equation}}
&\overset{\mathclap{(\alph{equation})}}{=} \sum_{m \in \mathcal{A}_s} \mathbb{E}[(Y(t)+h) 1\{q(t)=m, q(t+h)=s\}] + \sum_{m \in \mathcal{A}_s^c} \mathbb{E}[(\Delta_i(t)+h) 1\{q(t)=m, q(t+h)=s\}] + o(h) \notag \\ \noalign{\refstepcounter{equation}}
&\overset{\mathclap{(\alph{equation})}}{=} \sum_{m \in \mathcal{A}_s} \mathbb{E}(Y(t) 1\{q(t)=m, q(t+h)=s\}) + \sum_{m \in \mathcal{A}_s^c} \mathbb{E}(\Delta_i(t) 1\{q(t)=m, q(t+h)=s\}) \notag\\
&\quad + \sum_{m} h \cdot P\{q(t)=m, q(t+h)=s\} + o(h) \notag \\ \noalign{\refstepcounter{equation}}
&\overset{\mathclap{(\alph{equation})}}{=} \sum_{m \in \mathcal{A}_s} \mathbb{E}\Big(\mathbb{E}(Y(t) 1\{q(t)=m\} \cdot 1\{q(t+h)=s\} \mid q(t))\Big) \notag \\ 
&+ \sum_{m \in \mathcal{A}_s^c} \mathbb{E}\Big( \mathbb{E}(\Delta_i(t) 1\{q(t)=m\} \cdot 1\{q(t+h)=s\})|q(t)\Big) + h \cdot P\{q(t+h)=s\} + o(h).  \notag
\end{align}
\end{subequations}
Equation (\ref{iterated}$b$) inserts the identity $\sum_m 1\{q(t)=m\} = 1$.  Equation (\ref{iterated}$d$) captures two cases: for $m \in S$, a completion occurs, so $\Delta_i(t+h) = Y(t+h)$; for $m \in S^c$, no completion occurs and the AoI simply ages. Since the AoI grows deterministically at rate~$1$ between events, it follows that $Y(t+h) = Y(t) + h$. Moreover, we have exactly $\Delta_i(t+h) = \Delta_i(t) + h$ on any interval containing no service completion. Equation (\ref{iterated}$f$) expands the $(+h)$ terms and combines them into a single sum $h \sum_{m} P(q(t)=m,\, q(t+h)=s)$. Equation (\ref{iterated}$g$) then applies iterated expectations to the $m \in \mathcal{A}_s$ and $m \in \mathcal{A}_s^{c}$ terms with inner expectation conditioned on $q(t)$. Since the future queue state $1\{q(t+h)\}$ depends only on the current queue state $q(t)$, and not on the elapsed service time $Y(t)$ or the current AoI $\Delta_i(t)$, it is conditionally independent of both $Y(t)1\{q(t)=m\}$ and $\Delta_i(t)1\{q(t)=m\}$ given $q(t)$. This conditional independence is applied in (\ref{eq:vector_form}$a$).
\begin{subequations}
\refstepcounter{equation}
\begin{align} \label{eq:vector_form}
a_{i,\sigma(s)}&(t+h) \overset{\mathclap{(\alph{equation})}}{=} \sum_{m \in \mathcal{A}_s} \mathbb{E}\bigl(Y(t) 1\{q(t)=m\}\bigr) \mathbb{E}\bigl(1\{q(t+h)=s\} \mid q(t)=m)\bigr) \tag{\theparentequation} \\ 
&+ \sum_{m \in \mathcal{A}_s^c} \mathbb{E}\bigl(\Delta_i(t) 1\{q(t)=m\}\bigr) \mathbb{E}\bigl(1\{q(t+h)=s\} \mid q(t)=m)\bigr) + h \cdot P\{q(t+h)=s\} + o(h)  \notag \\ \noalign{\refstepcounter{equation}}
&\overset{\mathclap{(\alph{equation})}}{=} \sum_{m \in \mathcal{A}_s} \mathbb{E}\bigl(Y(t) 1\{q(t)=m\}\bigr) P\{q(t+h)=s \mid q(t)=m\} \notag\\
&+ \sum_{m \in \mathcal{A}_s^c} \mathbb{E}\bigl(\Delta_i(t) 1\{q(t)=m\}\bigr) P\{q(t+h)=s \mid q(t)=m\} + h \cdot P\{q(t+h)=s\} + o(h) \notag \\ \noalign{\refstepcounter{equation}}
&\overset{\mathclap{(\alph{equation})}}{=} \sum_{m \in \mathcal{A}_s} y_{\sigma(m)}(t) P\{q(t+h)=s \mid q(t)=m\} + \sum_{m} a_{i,\sigma(m)}(t) P\{q(t+h)=s \mid q(t)=m\} \notag\\
& - \sum_{m \in \mathcal{A}_s} a_{i,\sigma(m)}(t) P\{q(t+h)=s \mid q(t)=m\} + h \cdot P\{q(t+h)=s\} + o(h) \notag \\ \noalign{\refstepcounter{equation}}
&\overset{\mathclap{(\alph{equation})}}{=} \sum_{m \in \mathcal{A}_s} (y_{\sigma(m)}(t) - a_{i,\sigma(m)}(t)) P\{q(t+h)=s \mid q(t)=m\} \notag \\
&+ \sum_{m} a_{i,\sigma(m)}(t) \Big[ P\{q(t+h)=s \mid q(t)=m\} - 1\{m=s\} \Big] \notag \\ 
&+ \sum_{m} a_{i,\sigma(m)}(t) 1\{m=s\} + h \cdot P\{q(t+h)=s\} + o(h) \notag \\ \noalign{\refstepcounter{equation}}
&\overset{\mathclap{(\alph{equation})}}{=} \sum_{m \in \mathcal{A}_s} (\mathbf{y}(t) - \mathbf{a_i}(t))_{\sigma(m)} Q_{\sigma(m),\sigma(s)}(t) h + \sum_{m} (\mathbf{a}_i(t))_{\sigma(m)} Q_{\sigma(m),\sigma(s)}(t) h \notag \\
&+ (\mathbf{a}_i(t))_{\sigma(s)} + h \cdot (\mathbf{p}(t+h))_{\sigma(s)} + o(h). \notag
\end{align}
\end{subequations}
In (\ref{eq:vector_form}$c$), the sum over $m \in \mathcal{A}_s^c$ is rewritten as the full sum over all $m$ minus the sum restricted to $m \in \mathcal{A}_s$. We express transitions using the the generator matrix $\mathbf{Q}(t)$. For $m \in \mathcal{A}_s$, note that $m \neq s$, i.e. a completion changes the state. Equation (\ref{eq:vector_form}$e$) collects 
terms into vector notation: $(\mathbf{y}(t))_{\sigma(m)} = y_{\sigma(m)}(t)$ and $(\mathbf{a}_i(t))_{\sigma(m)} = a_{i,\sigma(m)}(t)$. It also utilizes the following fact: $P\{q(t+h)=s \mid q(t)=m\} - 1\{m=s\} = Q_{\sigma(m),\sigma(s)}(t)h + o(h)$. 
Collecting all components into the vector $\mathbf{a}_i(t)$:
\begin{align}
\mathbf{a}_i(t+h) &= (\mathbf{y}(t) - \mathbf{a}_i(t)) \mathbf{M}_i(t) h + \mathbf{a}_i(t) \mathbf{Q}(t) h + \mathbf{a}_i(t) + \mathbf{p}(t+h) h + o(h).\label{eq:vector_ode}
\end{align}
Here recall that $\mathbf{M}_i(t)$ is the subgenerator associated with the transitions triggered by class-$i$ completions, and $\mathbf{p}(t+h)$ is the state probability vector at $t+h$. Subtracting $\mathbf{a}_i(t)$ from both sides, dividing by $h$, and letting $h \to 0$:
\begin{align}
\dot{\mathbf{a}}_i(t) &= \lim_{h \to 0} \frac{\mathbf{a}_i(t+h) - \mathbf{a}_i(t)}{h} \notag\\
&= \lim_{h \to 0} \frac{1}{h} \Bigl[ (\mathbf{y}(t) - \mathbf{a}_i(t)) \mathbf{M}_i(t) h + \mathbf{a}_i(t) \mathbf{Q}(t) h + \mathbf{a}_i(t) + \mathbf{p}(t+h) h + o(h) - \mathbf{a}_i(t) \Bigr] \notag \\
&= \mathbf{p}(t) + \mathbf{a}_i(t) \mathbf{Q}(t) +  (\mathbf{y}(t) - \mathbf{a}_i(t)) \mathbf{M}_i(t). \notag
\end{align}
\begin{equation*}
\therefore \quad \boxed{\dot{\mathbf{a}}_i(t) = \mathbf{p}(t) + \mathbf{a}_i(t) \mathbf{Q}(t) + (\mathbf{y}(t) - \mathbf{a}_i(t)) \mathbf{M}_i(t)} 
\end{equation*}

\subsection{In-Service Packet Age Moments}

We define $D_i$ as the set of states from which a transition to a given state $s$ occurs with positive rate when a packet from class $i$ enters service. Specifically, we say that $m \in D_i$ if
\begin{equation}
\lim_{h \rightarrow 0} \frac {P(q(t)=m, q(t+h)=s \mid \mathcal{E}(t, t+h), B_i(q(t))=1, B_i(q(t+h))=0)} {h} > 0,
\end{equation}
for a given state $s$ under consideration where $\mathcal{E}(t,t+h)$ denotes the event that a completion occurs in $(t,t+h)$. 
\begin{subequations}
\refstepcounter{equation}
\begin{align} \label{eq:vanishing_term}
y_{\sigma(s)}&(t+h) \overset{\mathclap{(\alph{equation})}}{=} \mathbb{E}(Y(t+h)1\{q(t+h)=s\}) \tag{\theparentequation} \\ \noalign{\refstepcounter{equation}}
&\overset{\mathclap{(\alph{equation})}}{=} \mathbb{E}\Big(Y(t+h)1\{q(t+h)=s\}\sum_{m}1\{q(t)=m\}\Big) \notag \\ \noalign{\refstepcounter{equation}}
&\overset{\mathclap{(\alph{equation})}}{=} \sum_{m \in \mathcal{A}_s^c}\mathbb{E}(Y(t+h)1\{q(t)=m, q(t+h)=s\}) + \sum_{j=1}^{N}\sum_{m \in D_j}\mathbb{E}(Y(t+h)1\{q(t)=m, q(t+h)=s\}) \notag \\ \noalign{\refstepcounter{equation}}
&\overset{\mathclap{(\alph{equation})}}{=} \sum_{m \in \mathcal{A}_s^c}\mathbb{E}\Big((Y(t)+h \cdot 1\{J(s) \neq 0\}) \cdot 1\{q(t)=m, q(t+h)=s\}\Big) \notag \\
&+ \sum_{j=1}^{N}\sum_{m \in D_j}E(Z_j(t+h) 1\{q(t)=m, q(t+h)=s\}) + o(h) \notag \\ \noalign{\refstepcounter{equation}}
&\overset{\mathclap{(\alph{equation})}}{=} \sum_{m \in \mathcal{A}_s^c}\mathbb{E}\Big( \mathbb{E}\big( Y(t) 1\{q(t)=m\} 1\{q(t+h)=s\} \mid q(t) \big) \Big) \notag \\
&+ \sum_{m \in \mathcal{A}_s^c}h \cdot 1\{J(s) \neq 0\} \, \mathbb{E}(1\{q(t)=m, q(t+h)=s\}) \notag \\
&+ \sum_{m \in D_{j_0}}\mathbb{E}\Big( \mathbb{E}\big( Z_{j_0}(t+h) 1\{q(t)=m\} 1\{q(t+h)=s\} \mid q(t) \big) \Big) + o(h) \notag \\ \noalign{\refstepcounter{equation}}
&\overset{\mathclap{(\alph{equation})}}{=} \sum_{m \in \mathcal{A}_s^c} \mathbb{E}\Big(Y(t) 1\{q(t)=m\} \, P(q(t+h)=s \mid q(t))\Big) + \sum_{m} h \cdot 1\{J(s) \neq 0\} \, P\{q(t)=m, q(t+h)=s\} \notag \\ 
&+ \sum_{m \in D_{j_0}} \mathbb{E}\Big(Z_{j_0}(t+h) 1\{q(t)=m\} \, P\{q(t+h)=s \mid q(t)\}\Big) + o(h)
\notag \\ \noalign{\refstepcounter{equation}}
&\overset{\mathclap{(\alph{equation})}}{=} \sum_{m \in \mathcal{A}_s^c} y_{\sigma(m)}(t) \, P(q(t+h)=s \mid q(t)=m) + h \cdot 1\{J(s) \neq 0\} \, P(q(t+h)=s) \notag \\ 
&\quad + \sum_{m \in D_{j_0}} E\Big((Z_{j_0}(t)+h) 1\{q(t)=m\}\Big) \, P(q(t+h)=s \mid q(t)=m) + o(h) \notag \\ \noalign{\refstepcounter{equation}}
&\overset{\mathclap{(\alph{equation})}}{=} \sum_{m} y_{\sigma(m)}(t) \Big[ P\{q(t+h)=s \mid q(t)=m\} - 1\{m=s\} \Big] + \sum_{m} y_{\sigma(m)}(t) 1\{m=s\} \notag \\
&- \sum_{m \in \mathcal{A}_s} y_{\sigma(m)}(t) P\{q(t+h)=s \mid q(t)=m\} + h \cdot 1\{J(s) \neq 0\} \, P\{q(t+h)=s\} \notag \\ 
&+ \sum_{m \in D_{j_0}} z_{j_0,{\sigma(m)}}(t) \, P\{q(t+h)=s \mid q(t)=m\} \notag \\ 
&+ \underbrace{\sum_{m \in D_{j_0}} h \cdot P\{q(t)=m\} \, P\{q(t+h)=s \mid q(t)=m\}}_{o(h)} + o(h). \notag
\end{align}
\end{subequations}
We separate the sum into two parts in (\ref{eq:vanishing_term}$c$): states where the age accumulates smoothly over time, and states where a job completion causes an instantaneous jump in the age. Moreover, for a given target state $s$, at most one class $j_0$ can have $D_{j_0} \neq \emptyset$ (since a transition to state $s$ can only occur through at most one specific class entering service). Note that $J(s) = 0$ if and only if the system is in the idle state. The $o(h)$ terms arise from the probability of multiple events in $[t, t+h]$, which is negligible as $h \to 0$. Since $Z_{j_0}(t+h) 1\{q(t)=m\}$ and $1\{q(t+h)=s\}$ are conditionally independent given $q(t)$, the expectation factorizes in (\ref{eq:vanishing_term}$f$). The remaining term $E(1\{q(t+h)=s\} \mid q(t))$ is exactly the transition probability $P(q(t+h)=s \mid q(t))$. Since $Z_{j_0}$ increases at rate $1$ while in states of class $j_0$, we have $Z_{j_0}(t+h) = Z_{j_0}(t) + h$ in (\ref{eq:vanishing_term}$g$).
Note that the term with $h \cdot P\{q(t+h)=s \mid q(t)=m\}$ is $O(h^2)$, hence absorbed into $o(h)$ in \eqref{eq:vanishing_term}. Recall that $P\{q(t+h)=s \mid q(t)=m\} - 1\{m=s\} = Q_{m,s}(t)h + o(h)$. 
$\mathbf{M}^{(\mathrm{comp})}(t)$ denotes the matrix of transition rates due to completions only. Then:
\begin{equation}
\sum_{m} y_{\sigma(m)}(t) \Big[ P\{q(t+h)=s \mid q(t)=m\} - 1\{m=s\} \Big] = (\mathbf{y}(t) \mathbf{Q}(t))_{\sigma(s)} h + o(h).
\end{equation}
For the terms involving $z_{j_0}$:
\begin{align}
&\sum_{m \in D_{j_0}} z_{j_0,{\sigma(m)}}(t) \, P(q(t+h)=s \mid q(t)=m) \notag \\
& \quad = \sum_{m} (\mathbf{z}_{j_0}(t) \odot \boldsymbol{\delta}_{\text{next}=j_0})_{\sigma(m)} (\mathbf{M}^{(\mathrm{comp})}(t))_{{\sigma(m)},{\sigma(s)}} h + o(h) \notag \\
& \quad = \Big( (\mathbf{z}_{j_0}(t) \odot \boldsymbol{\delta}_{\text{next}=j_0}) (\mathbf{M}^{(\mathrm{comp})}(t)) \Big)_{\sigma(s)} h + o(h),
\label{eq:z_next}
\end{align}
where $\odot$ denotes element-wise product, and $\boldsymbol{\delta}_{\text{next}=j_0}$ is an indicator vector for class $j_0$. Notice that we can add terms for $j \neq j_0$ with zero contribution since $D_j = \emptyset$ for those classes. In other words, since $D_j = \emptyset$ for $j \neq j_0$, there is no state $m$ from which a class-$j$ completion leads to the state $s$, meaning the corresponding entries of column $\sigma(s)$ of $M^{\text{comp}}(t)$ are zero for all rows selected by $\boldsymbol{\delta}_{\text{next}=j}$. Therefore the product $(\mathbf{z}_j(t) \odot \boldsymbol{\delta}_{\text{next}=j}) M^{\text{comp}}(t)$ vanishes at index $\sigma(s)$ when $j \neq j_0$:
\begin{equation*}
\left( \Big(\sum_{\substack{j=1 \\ j \neq j_0}}^N \mathbf{z}_j(t) \odot \boldsymbol{\delta}_{\text{next}=j} \Big) \mathbf{M}^{(\mathrm{comp})}(t) \right)_{\sigma(s)} = 0.
\end{equation*}
Then, we sum over \emph{all} classes $j$ and define:
\begin{equation*}
\mathbf{z}^{\text{next}}(t) = \sum_{j=1}^N \mathbf{z}_j(t) \odot \boldsymbol{\delta}_{\text{next}=j}.
\end{equation*}
We express \eqref{eq:z_next} as follows:
\begin{equation*}
    \Big( (\mathbf{z}_{j_0}(t) \odot \boldsymbol{\delta}_{\text{next}=j_0}) 
    \mathbf{M}^{(\mathrm{comp})}(t) \Big)_{\sigma(s)} h + o(h) 
    = \Big( \mathbf{z}^{\text{next}}(t)\, \mathbf{M}^{(\mathrm{comp})}(t) \Big)_{\sigma(s)} h + o(h).
\end{equation*}
We now write the expression in vector-matrix form.
\begin{align*}
y_s(t+h) 
&= (\mathbf{y}(t) \mathbf{Q}(t))_{\sigma(s)} h + y_{\sigma(s)}(t) - (\mathbf{y}(t) \mathbf{M}^{(\mathrm{comp})}(t))_{\sigma(s)} h \notag \\
&\quad + h \cdot 1\{J(s) \neq 0\} \, (\mathbf{p}(t+h))_{\sigma(s)} + (\mathbf{z}^{\text{next}}(t) \mathbf{M}^{(\mathrm{comp})}(t))_{\sigma(s)} h + o(h).
\end{align*}
Recall that  $\boldsymbol{\delta}_{J\neq 0}$ is the indicator vector with $(\boldsymbol{\delta}_{J\neq 0})_{\sigma(s)} = 1\{J(s) \neq 0\}$ (1 if state $s$ is not idle). Then:
\begin{align*}
y_{\sigma(s)}(t+h) 
&= (\mathbf{y}(t) \mathbf{Q}(t))_{\sigma(s)} h + y_{\sigma(s)}(t) - (\mathbf{y}(t) \mathbf{M}_{\text{comp}}(t))_{\sigma(s)} h \notag \\
&\quad + h (\boldsymbol{\delta}_{J\neq 0} \odot \mathbf{p}(t+h))_{\sigma(s)} + (\mathbf{z}^{\text{next}}(t) \mathbf{M}^{(\mathrm{comp})}(t))_{\sigma(s)} h + o(h).
\end{align*}
Stacking over all states $s$, we have:
\begin{align*}
\mathbf{y}(t+h) &= \mathbf{y}(t) \mathbf{Q}(t) h + \mathbf{y}(t) - \mathbf{y}(t) \mathbf{M}^{(\mathrm{comp})}(t) h \notag \\
&\quad + \mathbf{p}(t+h) \odot \boldsymbol{\delta}_{J\neq 0} \cdot h + \mathbf{z}^{\text{next}}(t) \mathbf{M}_{\text{comp}}(t) h + o(h).
\end{align*}
Subtract $\mathbf{y}(t)$ from both sides, divide by $h$, and let $h \to 0$:
\begin{align*}
\dot{\mathbf{y}}(t) 
&= \lim_{h \to 0} \frac{\mathbf{y}(t+h) - \mathbf{y}(t)}{h} \notag \\[1ex] 
&= \lim_{h \to 0} \frac{1}{h} \Big[ \mathbf{y}(t)\mathbf{Q}(t)h - \mathbf{y}(t)\mathbf{M}_{\text{comp}}(t)h + \mathbf{p}(t+h) \odot \boldsymbol{\delta}_{J\neq 0} \cdot h + \mathbf{z}^{\text{next}}(t)\mathbf{M}^{(\mathrm{comp})}(t)h + o(h) \Big] \notag \\[1ex] 
&= \mathbf{p}(t) \odot \boldsymbol{\delta}_{J\neq 0} + \mathbf{y}(t)\mathbf{Q}(t) + (\mathbf{z}^{\text{next}}(t) - \mathbf{y}(t)) \mathbf{M}^{(\mathrm{comp})}(t).
\end{align*}
\begin{equation*}
\therefore \quad \boxed{\dot{\mathbf{y}}(t) = \mathbf{p}(t) \odot \boldsymbol{\delta}_{J\neq 0} + \mathbf{y}(t)\mathbf{Q}(t) + (\mathbf{z}^{\text{next}}(t) - \mathbf{y}(t)) \mathbf{M}^{(\mathrm{comp})}(t)} 
\end{equation*}

\subsection{Waiting Packet Age Moments}

We use the following notation throughout: $N_i(t, t+h)$ denotes the number of arrivals to class $i$ in the interval $(t, t+h)$. Recall that $B_i(s)$ is an indicator function equal to 1 if class $i$ is active in state $s$ and 0 otherwise. The derivation is decomposed into two cases according to the value of $B_i(s)$.
\subsubsection{Case 1: $B_i(s)=1$} 
 We begin by conditioning on the number of arrivals in a small interval $(t, t+h)$.
 \begin{subequations}
\refstepcounter{equation}
\begin{align} \label{eq:conditional_indep}
z_{i,s}&(t+h) \overset{\mathclap{(\alph{equation})}}{=} \mathbb{E} \Big(Z_i(t+h) 1\{q(t+h)=s\} \Big) \tag{\theparentequation} \\ \noalign{\refstepcounter{equation}} 
&\overset{\mathclap{(\alph{equation})}}{=} \mathbb{E} \Big(Z_i(t+h) 1\{q(t+h)=s, N_i(t, t+h)=0\} \Big) \notag \\ 
&+ \sum_{n=1}^{\infty} \underbrace{\mathbb{E} \Big(Z_i(t+h) 1\{q(t+h)=s, N_i(t, t+h)=n\} \mid N_i=n \Big)}_{\mathcal{O}(h)} \cdot \underbrace{P\{N_i(t, t+h)=n\}}_{\mathcal{O}(h)} + o(h) \notag \\ \noalign{\refstepcounter{equation}}
&\overset{\mathclap{(\alph{equation})}}{=} \mathbb{E} \Big(Z_i(t+h) 1\{q(t+h)=s, N_i(t, t+h)=0\} \Big) + o(h).  \notag \\ \noalign{\refstepcounter{equation}}
&\overset{\mathclap{(\alph{equation})}}{=} \mathbb{E} \Big( (Z_i(t)+h) 1\{q(t+h)=s, N_i(t, t+h)=0\} \Big) + o(h) \notag \\  \noalign{\refstepcounter{equation}}
&\overset{\mathclap{(\alph{equation})}}{=} \mathbb{E} \Big(Z_i(t) 1\{q(t+h)=s, N_i(t, t+h)=0\} \Big) + h \cdot P\{q(t+h)=s, N_i(t, t+h)=0\} + o(h) \notag \\ \noalign{\refstepcounter{equation}}
&\overset{\mathclap{(\alph{equation})}}{=} \sum_m \mathbb{E}\Big(Z_i(t) 1\{q(t+h)=s, N_i(t, t+h)=0\} \mid q(t)=m\Big) \cdot P\{q(t)=m\} \notag \\
&+ h \sum_n P\{q(t+h)=s, N_i(t, t+h)=0 \mid q(t)=n\} \cdot P\{q(t)=n\} + o(h) \notag \\ \noalign{\refstepcounter{equation}}
&\overset{\mathclap{(\alph{equation})}}{=} \sum_m\mathbb{E}\Big(Z_i(t) \mid q(t)=m\Big) P\{q(t)=m\} \cdot \mathbb{E}\Big(1\{q(t+h)=s, N_i(t, t+h)=0\} \mid q(t)=m\Big) \notag \\ 
&+ h \sum_n P\{q(t+h)=s, N_i(t, t+h)=0 \mid q(t)=n\} \cdot P\{q(t)=n\} + o(h). \notag
\end{align}
\end{subequations}
The key observation in (\ref{eq:conditional_indep}$b$) is that the conditional expectation is $\mathcal{O}(h)$ (since $Z_i$ can increase by at most $h$ in the interval) and the probability of one or more arrivals is also $\mathcal{O}(h)$. Their product is therefore $\mathcal{O}(h^2)$, which is absorbed into $o(h)$. This leaves only the no-arrival case in (\ref{eq:conditional_indep}$c$). When no arrival occurs to class $i$ in $(t, t+h)$, since no service completion occurs, $Z_i(t+h) = Z_i(t) + h$ simply increases by $h$. We apply the law of total expectation by conditioning on the state at time $t$ in (\ref{eq:conditional_indep}$f$). Since $Z_i(t)$ depends only on the history up to time $t$, and the evolution after $t$ is conditionally independent given $q(t)=m$, we can factor the conditional expectation in (\ref{eq:conditional_indep}$g$). Recognizing that $z_{i,\sigma(m)}(t) = \mathbb{E}[Z_i(t) \mid q(t)=m] P\{q(t)=m\}$ by definition, we obtain:
\begin{subequations}
\refstepcounter{equation}
\begin{align} \label{eq:B_i_0_vanish}
z_{i,\sigma(s)}&(t+h) \overset{\mathclap{(\alph{equation})}}{=} \sum_m z_{i,\sigma(m)}(t) P\{q(t+h)=s, N_i(t, t+h)=0 \mid q(t)=m\} \tag{\theparentequation} \\  
& + h \sum_n P\{q(t+h)=s, N_i(t, t+h)=0 \mid  q(t)=n\} \cdot P\{q(t)=n\} + o(h). \notag \\ \noalign{\refstepcounter{equation}}
&\overset{\mathclap{(\alph{equation})}}{=} \sum_{\substack{m \\ B_i(m)=1}} z_{i,\sigma(m)}(t) P\{q(t+h)=s \mid q(t)=m\} \cdot P\{N_i(t, t+h)=0 \mid q(t)=m\} \notag \\ 
&+ h \sum_{\substack{n \\ B_i(n)=1}} P\{q(t+h)=s \mid q(t)=n\} \cdot P\{N_i(t, t+h)=0 \mid q(t)=n\} P\{q(t)=n\} + o(h)  \notag \\ \noalign{\refstepcounter{equation}}
&\overset{\mathclap{(\alph{equation})}}{=} \sum_{\substack{m \\ B_i(m)=1}} z_{i,\sigma(m)}(t) P\{q(t+h)=s \mid q(t)=m\} \notag \cdot P\{N_i(t, t+h)=0\} \notag \\
& + h \Bigg[ \sum_{\substack{n \\ B_i(n)=1}} \Big[ \underbrace{P\{q(t+h)=s \mid q(t)=n\} - 1\{n=s\}}_{\mathcal{O}(h)} \Big] \cdot P\{N_i(t, t+h)=0\} P\{q(t)=n\} \notag \\
& + \sum_{\substack{n \\ B_i(n)=1}} 1\{n=s\} P\{N_i(t, t+h)=0\} P\{q(t)=n\} \Bigg] + o(h) \notag \\ \noalign{\refstepcounter{equation}}
&\overset{\mathclap{(\alph{equation})}}{=} \sum_{\substack{m \\ B_i(m)=1}} z_{i,\sigma(m)}(t) P\{q(t+h)=s \mid q(t)=m\} \cdot P\{N_i(t, t+h)=0\} + h \cdot (1 - \lambda_i(t) h) P\{q(t)=s\} + o(h).
\notag
\end{align}
\end{subequations}
We split the sums based on whether class $i$ is active in state $m$ in (\ref{eq:B_i_0_vanish}$b$). When $B_i(m)=0$, we have $z_{i,m}(t)=0$ because no time is accumulated in states where class $i$ is inactive. This eliminates the second group of terms. Note that the dependence between the events $\{q(t+h)=s\}$ and $\{N_i(t,t+h)\}$ is negligible to first order in $h$; hence, their joint probability is well approximated by the product of the respective conditional probabilities, with any residual coupling absorbed into the $o(h)$ term in (\ref{eq:B_i_0_vanish}$b$). In the the second term of (\ref{eq:B_i_0_vanish}$b$), we expand the transition probability by adding and subtracting the indicator $1\{n=s\}$. The bracketed difference is $\mathcal{O}(h)$ by definition of the generator, and when multiplied by the outer $h$ it becomes $\mathcal{O}(h^2)$ and is absorbed into $o(h)$. Hence, only $1\{n=s\}$ term remains, yielding $hP\{N_i(t,t+h)=0\}P\{q(t)=s\} = h(1-\lambda_i(t)h)P\{q(t)=s\}$. We replace the conditional no-arrival probability $P\{N_i(t,t+h)=0 \mid q(t)=m\}$ with the unconditional probability $P\{N_i(t,t+h)=0\}$, since the event $N_i(t,t+h)=0$ is independent of the event $q(t)=m$.
\begin{subequations}
\refstepcounter{equation}
\begin{align} \label{eq:generator}
z_{i,\sigma(s)}&(t+h) \overset{\mathclap{(\alph{equation})}}{=} \sum_{\substack{m \\ B_i(m)=1}} z_{i,\sigma(m)}(t) P\{q(t+h)=s \mid q(t)=m\} \cdot P\{N_i(t, t+h)=0\} + h P\{q(t)=s\} + o(h) \tag{\theparentequation} \\ \noalign{\refstepcounter{equation}}
&\overset{\mathclap{(\alph{equation})}}{=}\sum_{\substack{m \\ B_i(m)=1}} z_{i,\sigma(m)}(t) P\{q(t+h)=s \mid q(t)=m\} \cdot (1 - \lambda_i(t)h) + h P\{q(t)=s\} + o(h) \notag \\ \noalign{\refstepcounter{equation}}
&\overset{\mathclap{(\alph{equation})}}{=} \sum_{\substack{m \\ B_i(m)=1}} z_{i,\sigma(m)}(t) \Bigg[ \Big( P\{q(t+h)=s \mid q(t)=m\} - 1\{m=s\} \Big) + 1\{m=s\} \Bigg] (1 - \lambda_i(t)h) \notag \\
&+ h P\{q(t)=s\} + o(h) \notag \\ \noalign{\refstepcounter{equation}}
&\overset{\mathclap{(\alph{equation})}}{=} \sum_{\substack{m \\ B_i(m)=1}} z_{i,\sigma(m)}(t) \Big( P\{q(t+h)=s \mid q(t)=m\} - 1\{m=s\} \Big) + \sum_{\substack{m \\ B_i(m)=1}} z_{i,\sigma(m)}(t) 1\{m=s\} \notag \\
&- \lambda_i(t)h \sum_{\substack{m \\ B_i(m)=1}} z_{i,\sigma(m)}(t) \underbrace{\Big[ P\{q(t+h)=s \mid q(t)=m\} - 1\{m=s\} \Big]}_{\mathcal{O}(h)} \notag \\
&- \lambda_i(t)h \sum_{\substack{m \\ B_i(m)=1}} z_{i,\sigma(m)}(t) 1\{m=s\} + h P\{q(t)=s\} + o(h). \notag 
\end{align}
\end{subequations}
We expand the transition probability in the first sum by adding and subtracting $1\{m=s\}$, then distribute the factor $(1-\lambda_i(t)h)$ across the two resulting terms in (\ref{eq:generator}$c$). The term with the bracketed $\mathcal{O}(h)$ difference multiplied by $h$ is $\mathcal{O}(h^2)$ and is absorbed into $o(h)$. The indicator $1\{m=s\}$ collapses the second and fourth sums to $m=s$, giving $z_{i,s}(t)$ and $-\lambda_i(t)h z_{i,s}(t)$ respectively. Recall that $B_i(s) = 1$ in this case. The difference $P\{q(t+h)=s \mid q(t)=m\} - 1\{m=s\}$ in the first sum in (\ref{eq:generator}$d$) equals $Q_{m,s}(t)h$ by the definition of the generator matrix.
\begin{align}
z_{i,\sigma(s)}&(t+h) = \sum_{\substack{m \\ B_i(m)=1}} z_{i,\sigma(m)}(t) Q_{\sigma(m),\sigma(s)}(t)h + z_{i,\sigma(s)}(t) - \lambda_i(t) z_{i,\sigma(s)}(t)h + h P\{q(t)=s\} + o(h). \notag \\
&= (\mathbf{z}_i(t) \mathbf{Q}(t))_s h + (\mathbf{z}_i(t))_{\sigma(s)} - \lambda_i(t) (\mathbf{z}_i(t))_{\sigma(s)} h + h (\mathbf{p}(t))_{\sigma(s)} + o(h). \label{eq:matrix_eq}
\end{align}
We identify the vector-matrix product $(\mathbf{z}_i(t) \mathbf{Q}(t))_{\sigma(s)} = \sum_m (\mathbf{z}_i(t))_{\sigma(m)} (\mathbf{Q}(t))_{\sigma(m),\sigma(s)}$, and write $(\mathbf{z}_i(t))_{\sigma(s)} = \mathbf{z}_{i,\sigma(s)}(t)$ and $(\mathbf{p}(t))_{\sigma(s)} = P\{q(t)=s\}$ in \eqref{eq:matrix_eq}.
Taking the limit as $h \to 0$, we obtain the time derivative:
\begin{align*}
\dot{z}_{i,\sigma(s)}(t) &= \lim_{h \to 0} \frac{z_{i,\sigma(s)}(t+h) - z_{i,\sigma(s)}(t)}{h} \notag \\ 
&= \lim_{h \to 0} \frac{1}{h} \Big[ (\mathbf{z}_i(t) \mathbf{Q}(t))_{\sigma(s)} h + z_{i,\sigma(s)}(t) - \lambda_i(t) (\mathbf{z}_i(t))_s{\sigma(s)} h + h (\mathbf{p}(t))_{\sigma(s)} + o(h) - z_{i,\sigma(s)}(t) \Big] \notag \\
&= (\mathbf{z}_i(t)\mathbf{Q}(t))_{\sigma(s)} - \lambda_i(t)(\mathbf{z}_i(t))_{\sigma(s)} + (\mathbf{p}(t))_{\sigma(s)}.
\end{align*}

\subsubsection{Case 2: $B_i(s)=0$} 

When class $i$ is not active in state $s$,  $Z_i(t)=0$ as long as the system remains in that state. Therefore:
\begin{align*}
\dot{z}_{i,s}(t) &= \lim_{h \to 0} \frac{z_{i,s}(t+h) - z_{i,s}(t)}{h} = 0.
\end{align*}

We combine both cases using the element-wise product with the indicator vector $\boldsymbol{\delta}_{B_i = 1}$ where $(\boldsymbol{\delta}_{B_i = 1})_s = 1\{B_i(s)=1\}$, we obtain the unified equation:

\begin{equation}
\dot{\mathbf{z}}_i(t) = \mathbf{p}(t) \odot \boldsymbol{\delta}_{B_i = 1} + (\mathbf{z}_i(t) \mathbf{Q}(t)) \odot \boldsymbol{\delta}_{B_i = 1} - \lambda_i(t)(\mathbf{z}_i(t) \odot \boldsymbol{\delta}_{B_i = 1}). \notag\\
\end{equation}
We now carry out the algebraic manipulations on the resulting ODE. We consider the $\sigma(s)$-th component and distinguish between the two possible values of $B_i(s)$. Consider the case when $B_i(s) = 0$.
\begin{align*}
\Big( \mathbf{z}_i(t)&\big( \mathbf{Q}(t) - \mathbf{M}^{(\mathrm{next}=i)}(t) \big) \Big)_{\sigma(s)} = \sum_q (\mathbf{z}_i(t))_{\sigma(q)} \Big( \mathbf{Q}(t) - \mathbf{M}^{(\mathrm{next}=i)}(t) \Big)_{\sigma(q),\sigma(s)} \\
&= \sum_{\substack{q \\ B_i(q)=0}} (\mathbf{z}_i(t))_{\sigma(q)} \Big( \mathbf{Q}(t) - \mathbf{M}^{(\mathrm{next}=i)}(t) \Big)_{\sigma(q),\sigma(s)} + \sum_{\substack{q \\ B_i(q)=1}} (\mathbf{z}_i(t))_{\sigma(q)} \Big( \mathbf{Q}(t) - \mathbf{M}^{(\mathrm{next}=i)}(t) \Big)_{\sigma(q),\sigma(s)}. 
\end{align*}
Recall that $\mathbf{M}^{(\mathrm{next}=i)}(t)$ is the matrix containing only those transition rates that lead to class $i$ becoming the next class in service. The first sum is zero because $z_{i,\sigma(q)}(t)=0$ whenever $B_i(q)=0$. For the second sum, since $B_i(q)=1$ and $B_i(s)=0$, the transition from state $q$ to state $s$ corresponds to an event that makes class $i$ become the next class in service. Therefore,
\begin{equation*}
(\mathbf{Q}(t))_{\sigma(q),\sigma(s)} = (\mathbf{M}^{(\mathrm{next}=i)}(t))_{\sigma(q),\sigma(s)},
\end{equation*}
and the difference vanishes componentwise. Hence,
\begin{equation*}
\Big( \mathbf{z}_i(t)\big( \mathbf{Q}(t) - \mathbf{M}^{(\mathrm{next}=i)}(t) \big) \Big)_{\sigma(s)} = 0.
\end{equation*}
This agrees with
\begin{equation*}
\Big( (\mathbf{z}_i(t)\mathbf{Q}(t)) \odot \boldsymbol{\delta}_{B_i = 1} \Big)_{\sigma(s)} = 0,
\end{equation*}
since the elementwise product keeps only the components for which $B_i(s)=1$. Now, consider the case when $B_i(s) = 1$.
In this case, the state $s$ already belongs to the active set of class $i$, so transitions into the state indexed by $s$ are not arrivals that trigger the next service of class $i$. Consequently, 
\begin{equation*}
(\mathbf{M}^{(\mathrm{next}=i)}(t))_{\sigma(q),\sigma(s)}=0 \qquad \text{for all } q.
\end{equation*}
Thus,
\begin{align*}
\Big( \mathbf{z}_i(t)\big( \mathbf{Q}(t) - \mathbf{M}^{(\mathrm{next}=i)}(t) \big) \Big)_{\sigma(s)}
&= \sum_q \big( \mathbf{z}_i(t) \big)_{\sigma(q)} \Big( \mathbf{Q}(t) - \mathbf{M}^{(\mathrm{next}=i)}(t) \Big)_{\sigma(q),\sigma(s)} \\
&= \sum_q (\mathbf{z}_i(t))_{\sigma(q)} (\mathbf{Q}(t))_{\sigma(q),\sigma(s)} \notag \\
&= (\mathbf{z}_i(t)\mathbf{Q}(t))_{\sigma(s)} \notag \\
&= \Big( (\mathbf{z}_i(t)\mathbf{Q}(t)) \odot \boldsymbol{\delta}_{B_i = 1} \Big)_{\sigma(s)},
\end{align*}
which is exactly the desired expression. Since the two cases agree for every component $s$, we conclude that
\begin{equation*}
\mathbf{z}_i(t)\big( \mathbf{Q}(t) - \mathbf{M}^{(\mathrm{next}=i)}(t) \big)
=
(\mathbf{z}_i(t)\mathbf{Q}(t)) \odot \boldsymbol{\delta}_{B_i = 1}.
\end{equation*}
Substituting this identity into the ODE gives
\begin{align*}
\dot{\mathbf{z}}_i(t)
&= \mathbf{p}(t)\odot \boldsymbol{\delta}_{B_i = 1}
+ \mathbf{z}_i(t)\big( \mathbf{Q}(t) - \mathbf{M}^{(\mathrm{next}=i)}(t) \big) - \lambda_i(t)\big( \mathbf{z}_i(t)\odot \boldsymbol{\delta}_{B_i = 1} \big) \notag \\
&= \mathbf{p}(t)\odot \boldsymbol{\delta}_{B_i = 1}
+ \mathbf{z}_i(t)\mathbf{Q}(t)
- \lambda_i(t)\big( \mathbf{z}_i(t)\odot \boldsymbol{\delta}_{B_i = 1} \big) - \mathbf{z}_i(t)\mathbf{M}^{(\mathrm{next}=i)}(t).
\end{align*}
\begin{equation*}
\therefore \quad \boxed{\dot{\mathbf{z}}_i(t) = \mathbf{p}(t)\odot \boldsymbol{\delta}_{B_i = 1}
+ \mathbf{z}_i(t)\mathbf{Q}(t)
- \lambda_i(t)\big( \mathbf{z}_i(t)\odot \boldsymbol{\delta}_{B_i = 1} \big) - \mathbf{z}_i(t)\mathbf{M}^{(\mathrm{next}=i)}(t)} 
\end{equation*}

\subsection{Mean AoI}

We decompose the instantaneous AoI $\Delta_i(t)$ over the countable state space of the underlying CTMC. 
\begin{subequations}
\refstepcounter{equation}
\begin{align} \label{eq:mean}
\bar{\Delta}_i(t) 
&\overset{\mathclap{(\alph{equation})}}{=} \mathbb{E}\left(\Delta_{i}(t)\sum_{s}1_{\{q(t)=s\}}\right) \tag{\theparentequation}  \\
\noalign{\refstepcounter{equation}}
&\overset{\mathclap{(\alph{equation})}}{=} \sum_{s}\mathbb{E}\bigl(\Delta_i(t)1_{\{q(t)=s\}}\bigr) \notag\\
\noalign{\refstepcounter{equation}}
&\overset{\mathclap{(\alph{equation})}}{=} \sum_{s}a_{i,\sigma(s)}(t) \notag \\
\noalign{\refstepcounter{equation}}
&\overset{\mathclap{(\alph{equation})}}{=} \mathbf{a_i(t)}\mathbf{1}. \notag
\end{align}
\end{subequations}
Equation (\ref{eq:mean}$b$) follows from the linearity of expectation and the fact that the indicator partitions the sample space. The resulting scalar $\bar{\Delta}_i(t)$ is therefore the sum of state-dependent moment functions across all feasible states.

\subsection{Mean PAoI} \label{peak_AOI}
Recall $\mathcal{E}_i(t,t+h)$ denotes the event that a class‑$i$ completion occurs in the infinitesimal interval $(t,t+h)$. The mean peak AoI $\hat{\Delta}_i(t)$ is defined as the limit of the conditional expectation:
\begin{subequations}
\refstepcounter{equation}
\begin{align} \label{eq:cond_prob_def}
\hat{\Delta}_i(t) 
&\overset{\mathclap{(\alph{equation})}}{=} \lim_{h\rightarrow0}\mathbb{E}\bigl(\Delta_i(t) \mid \mathcal{E}_i(t,t+h)\bigr) \tag{\theparentequation} \\ \noalign{\refstepcounter{equation}}
&\overset{\mathclap{(\alph{equation})}}{=} \lim_{h\rightarrow0}\mathbb{E}\bigl(\Delta_i(t)1\{\mathcal{E}_i(t,t+h)\} \mid \mathcal{E}_i(t,t+h)\bigr) \notag \\ \noalign{\refstepcounter{equation}}
&\overset{\mathclap{(\alph{equation})}}{=} \lim_{h\rightarrow0} \frac{\mathbb{E}\bigl(\Delta_i(t)1\{\mathcal{E}_i(t,t+h)\}\bigr)}{P\{\mathcal{E}_i(t,t+h)\}}. \notag
\end{align}
\end{subequations}
Equation (\ref{eq:cond_prob_def}$c$) applies the standard definition of conditional expectation for non‑zero probability events.  To evaluate the limit, we separately characterize the numerator and the denominator using the law of total probability over the CTMC states. 
\begin{subequations}
\refstepcounter{equation}
\begin{align} \label{eq:indep_service}
\mathbb{E}\bigl(\Delta_i(t)\mathbf{1}\{\mathcal{E}_i(t,t+h)\}\bigr) &\overset{\mathclap{(\alph{equation})}}{=} \sum_{s} \mathbb{E}\bigl(\Delta_i(t)\mathbf{1}\{\mathcal{E}_i(t,t+h)\} \mid q(t)=s\bigr)\,P\{q(t)=s\} \tag{\theparentequation} \\ \noalign{\refstepcounter{equation}}
&\overset{\mathclap{(\alph{equation})}}{=} \sum_{\substack{s \\ J(s)=i}} \mathbb{E}\bigl(\Delta_i(t) \mid q(t)=s\bigr)\,
\mathbb{E}\bigl(\mathbf{1}\{\mathcal{E}_i(t,t+h)\} \mid q(t)=s\bigr)\,P\{q(t)=s\} \notag \\ \noalign{\refstepcounter{equation}}
&\overset{\mathclap{(\alph{equation})}}{=} \sum_{\substack{s \\ J(s)=i}} \mathbb{E}\bigl(\Delta_i(t) \mid q(t)=s\bigr)\,P\{q(t)=s\}\,\mu_i(t)h + o(h) \notag\\ \noalign{\refstepcounter{equation}}
&\overset{\mathclap{(\alph{equation})}}{=} \sum_{\substack{s \\ J(s)=i}} \mathbb{E}\bigl(\Delta_i(t)\mathbf{1}\{q(t)=s\}\bigr)\,\mu_i(t)h + o(h) \notag\\ \noalign{\refstepcounter{equation}}
&\overset{\mathclap{(\alph{equation})}}{=} \sum_{\substack{s \\ J(s)=i}} a_{i,\sigma(s)}(t)\,\mu_i(t)h + o(h). \notag
\end{align}
\end{subequations}
In (\ref{eq:indep_service}$c$) we exploit the fact that, given $q(t)=s$, the residual service time is independent of the current AoI. Equation (\ref{eq:indep_service}$f$) substitutes the definition $a_{i,s}(t)\triangleq\mathbb{E}[\Delta_i(t)\mid q(t)=s]$.
Similarly, marginalizing over $q(t)$ yields:
\begin{align}
P\{\mathcal{E}_i(t,t+h)\} 
&= \sum_{s} P\{\mathcal{E}_i(t,t+h) \mid q(t)=s\}P\{q(t)=s\} \notag\\
&= \sum_{\substack{s \\ J(s)=i}} \mu_i(t)h\,P\{q(t)=s\} + o(h). \label{eq:denom_final}
\end{align}
Substituting (\ref{eq:indep_service}$f$) and~\eqref{eq:denom_final} into (\ref{eq:cond_prob_def}$c$) and canceling the common factor $h$:
\begin{subequations}
\refstepcounter{equation}
\begin{align} \label{eq:cancel_v}
\hat{\Delta}_i(t) 
&\overset{\mathclap{(\alph{equation})}}{=} \lim_{h\rightarrow 0}
\frac{\sum_{\substack{s \\ J(s)=i}} a_{i,\sigma(s)}(t)\mu_i(t)h + o(h)}
{\sum_{\substack{s \\ J(s)=i}} \mu_i(t)h\,P\{q(t)=s\} + o(h)} \tag{\theparentequation} \\ \noalign{\refstepcounter{equation}}
&\overset{\mathclap{(\alph{equation})}}{=} \lim_{h\rightarrow 0}
\frac{\sum_{\substack{s \\ J(s)=i}} a_{i,\sigma(s)}(t) \mu_i(t) + \dfrac{o(h)}{h}}
{\sum_{\substack{s \\ J(s)=i}} P\{q(t)=s\} \mu_i(t) + \dfrac{o(h)}{h}} \notag\\ \noalign{\refstepcounter{equation}}
&\overset{\mathclap{(\alph{equation})}}{=} \frac{\sum_{\substack{s \\ J(s)=i}} a_{i,\sigma(s)}(t)}
{\sum_{\substack{s \\ J(s)=i}} P\{q(t)=s\}} \notag \\
\noalign{\refstepcounter{equation}}
&\overset{\mathclap{(\alph{equation})}}{=} \frac{\mathbf{a}_i(t)\mathbf{1}_{J=i}}{\mathbf{p}(t)\mathbf{1}_{J=i}}. \notag
\end{align}
\end{subequations}
Equation (\ref{eq:cancel_v}$d$) follows because the service rate $\mu_i(t)$ is common to both numerator and denominator for all states where $J(s)=i$, and the $o(h)/h$ terms vanish as $h\rightarrow 0$.

\section{Proof of Theorem \ref{existence_thm_main}} \label{appendix_b}

This appendix details the convergence analysis of the system. We begin by reformulating the dynamics using the Reduced Kolmogorov Forward Equation and aggregating them into a compact linear ODE representation.We then establish the structural properties of the system to formally prove the existence, uniqueness, and exponential convergence of the periodic steady state.

\begin{proposition}[Reduced Kolmogorov Equation]
\label{prop_reduced}
The Kolmogorov Forward Equation $\dot{\mathbf{p}}(t) = \mathbf{p}(t)\,\mathbf{Q}(t)$ admits 
the reduced form:
\begin{equation*}
    \dot{\mathbf{p}}_{\mathrm{red}}(t) = 
    \mathbf{p}_{\mathrm{red}}(t)\,\mathbf{Q}_{\mathrm{red}}(t) + 
    \boldsymbol{\beta}(t),
\end{equation*}
where $\mathbf{Q}_{\mathrm{red}}(t)$ is the reduced generator 
matrix, $\mathbf{p_{red}}(t)$ is the reduced state probability vector, and $\boldsymbol{\beta}(t)$ is the corresponding source term.
\end{proposition}

\begin{proof}

We reduce the dimension of the state probability vector $\mathbf{p}(t)$ by using the normalization condition
\begin{equation}
\sum_{n=1}^{N} p_n(t)=1.
\end{equation}
Hence, the last component can be written in terms of the first $N-1$ components as
\begin{equation}
p_N(t)=1-\sum_{n=1}^{N-1} p_n(t).
\end{equation}
Therefore, the full probability vector can be expressed as
\begin{equation}
\mathbf{p}(t)=\big[ p_1(t)\;\; \dots \;\; p_N(t) \big]
= \bigg[ p_1(t)\;\; \dots \;\; 1-\sum_{n=1}^{N-1} p_n(t) \bigg].
\end{equation}
Equivalently, letting
\begin{equation}
\mathbf{p}_{\text{red}}(t)\triangleq \big[ p_1(t)\;\; \dots \;\; p_{N-1}(t) \big],
\end{equation}
we can write
\begin{equation}
\mathbf{p}(t)=\mathbf{p}_{\text{red}}(t)\mathbf{T}+\mathbf{b},
\end{equation}
where
\begin{equation}
\mathbf{T}=
\begin{bmatrix}
1 & 0 & \dots & -1 \\
0 & 1 & \dots & -1 \\
\vdots & \vdots & \ddots & \vdots \\
0 & 0 & \dots & -1
\end{bmatrix},
\qquad
\mathbf{b}=\big[ 0 \;\; \dots \;\; 0 \quad 1 \big].
\end{equation}
Note that
\begin{equation}
\mathbf{p}(t)\mathbf{Q}(t)
=
\big(\mathbf{p}_{\text{red}}(t)\mathbf{T}+\mathbf{b}\big)\mathbf{Q}(t)
=
\mathbf{p}_{\text{red}}(t)\mathbf{T}\mathbf{Q}(t)+\mathbf{b}\mathbf{Q}(t),
\end{equation}
Next, we keep only the first $N-1$ components by multiplying with 
\begin{equation}
\mathbf{E}=
\begin{bmatrix}
\mathbf{I}_{N-1} \\
\mathbf{0}^T
\end{bmatrix},
\end{equation}
Right multiplication by $\mathbf{E}$ extracts the reduced state components. Applying this projection gives
\begin{equation}
\mathbf{p}(t)\mathbf{Q}(t)\mathbf{E}
=
\mathbf{p}_{\text{red}}(t)\mathbf{T}\mathbf{Q}(t)\mathbf{E}
+
\mathbf{b}\mathbf{Q}(t)\mathbf{E}.
\end{equation}
Using the Kolmogorov Forward Equation, we obtain
\begin{equation}
\dot{\mathbf{p}}(t)\mathbf{E}
=
\mathbf{p}_{\text{red}}(t)\underbrace{\mathbf{T}\mathbf{Q}(t)\mathbf{E}}_{\mathbf{Q}_{\text{red}}(t)}
+
\underbrace{\mathbf{b}\mathbf{Q}(t)\mathbf{E}}_{\boldsymbol{\beta}(t)}.
\end{equation}
Thus, the final equation takes the compact form
\begin{equation}
\dot{\mathbf{p}}_{\text{red}}(t)
=
\mathbf{p}_{\text{red}}(t)\mathbf{Q}_{\text{red}}(t)
+
\boldsymbol{\beta}(t).
\end{equation}
\end{proof}

The following corollary consolidates Theorem~\ref{theorem_main} and Proposition~\ref{prop_reduced} into a single compact ODE.

\begin{corollary}[Compact ODE Representation]
\label{cor_compact}
Define the stacked vector:
\begin{equation}
\label{eq_stack}
    \mathbf{x}(t) = [\mathbf{a}_1(t), \dots, \mathbf{a}_N(t), \mathbf{y}(t), \mathbf{z}_1(t), \dots, \mathbf{z}_N(t), \mathbf{p}_{\text{red}}(t)]^T.
\end{equation}
Then the system of ODEs in Theorem~\ref{theorem_main} and 
Proposition~\ref{prop_reduced} admits the compact representation:
\begin{equation} 
\label{eq_compact}
    \dot{\mathbf{x}}(t) = \mathbf{A}(t)\,\mathbf{x}(t) + \mathbf{b}(t),
\end{equation}
where $\mathbf{b}(t) = [\boldsymbol{\beta}(t),\, \mathbf{0},\, 
\ldots,\, \mathbf{0}]^\top$ with $\boldsymbol{\beta}(t)$ as in 
Proposition~\ref{prop_reduced}, and $\mathbf{A}(t)$ is a block 
upper triangular matrix whose diagonal blocks are given by:
\begin{align}
    \mathbf{A}_{i,\,i}(t) &= \mathbf{Q}^\top(t) - 
    \mathbf{M}^{(i)\top}(t), \quad i = 1, \ldots, N, \\[4pt]
    \mathbf{A}_{N+1,\,N+1}(t) &= \mathbf{Q}^\top(t) - 
    \mathbf{M}^{(\mathrm{comp})\top}(t), \\[4pt]
    \mathbf{A}_{N+1+i,\,N+1+i}(t) &= \mathbf{Q}^\top(t) - 
    \mathbf{M}^{(\mathrm{next}=i)\top}(t) - 
    \lambda_i(t)\,\mathrm{diag}(\boldsymbol{\delta}_{B_i=1}), 
    \quad i = 1, \ldots, N, \\[4pt]
    \mathbf{A}_{2N+2,\,2N+2}(t) &= \mathbf{Q}_{\mathrm{red}}^\top(t).
\end{align}

where $\mathrm{diag}(\mathbf{v})$ denotes the diagonal matrix 
formed from the entries of vector $\mathbf{v}$. Note that all matrices are transposed since $\mathbf{x}(t)$ is 
defined as a column vector in \eqref{eq_stack}.
\end{corollary}

\begin{proof}

We begin by expressing the ODEs from Theorem 1 purely in terms of matrix multiplications. Noting the identity for the elementwise product $\mathbf{p}(t) \odot \mathbf{v} = \mathbf{p}(t)\text{diag}(\mathbf{v})$, the row-vector ODEs for the age moments can be rewritten as:
\begin{align}
    \dot{\mathbf{z}}_i(t) &= \mathbf{p}(t)\text{diag}(\boldsymbol{\delta}_{B_i = 1}) + \mathbf{z}_i(t)\big[\mathbf{Q}(t) - \mathbf{M}^{(\text{next}=i)}(t) - \lambda_i(t)\text{diag}(\boldsymbol{\delta}_{B_i = 1})\big], \label{eq:z_row} \\
    \dot{\mathbf{y}}(t) &= \mathbf{p}(t)\text{diag}(\boldsymbol{\delta}_{J\neq 0}) + \mathbf{y}(t)\big[\mathbf{Q}(t) - \mathbf{M}^{(\text{comp})}(t)\big] + \sum_{k=1}^N \mathbf{z}_k(t)\text{diag}(\boldsymbol{\delta}_{\text{next}=k}), \label{eq:y_row} \\
    \dot{\mathbf{a}}_i(t) &= \mathbf{p}(t) + \mathbf{y}(t) + \mathbf{a}_i(t)\big[\mathbf{Q}(t) - \mathbf{M}^{(i)}(t)\big]. \label{eq:a_row}
\end{align}

Alongside these age moment ODEs, we recall from Proposition 2 the Reduced Kolmogorov Forward Equation governing the reduced state probabilities:
\begin{equation}
    \dot{\mathbf{p}}_{\text{red}}(t) = \mathbf{p}_{\text{red}}(t) \mathbf{Q}_{\text{red}}(t) + \boldsymbol{\beta}(t). \label{eq:p_red_row}
\end{equation}

Corollary 1 defines the stacked state $\mathbf{x}(t)$ as a \textit{column} vector. To cast the system into the form $\dot{\mathbf{x}}(t) = \mathbf{A}(t)\mathbf{x}(t) + \mathbf{b}(t)$, we transpose equations \eqref{eq:z_row}--\eqref{eq:p_red_row}:
\begin{align}
    &\dot{\mathbf{p}}_{\text{red}}^T(t) = \mathbf{Q}_{\text{red}}^T(t)\mathbf{p}_{\text{red}}^T(t) + \boldsymbol{\beta}^T(t), \label{eq:source_term}\\
    &\dot{\mathbf{z}}_i^T(t) = \big[\mathbf{Q}^T(t) - \mathbf{M}^{(\text{next}=i)T}(t) - \lambda_i(t)\text{diag}(\boldsymbol{\delta}_{B_i = 1})\big]\mathbf{z}_i^T(t) + \text{diag}(\boldsymbol{\delta}_{B_i = 1})\mathbf{p}^T(t), \\
    &\dot{\mathbf{y}}^T(t) = \big[\mathbf{Q}^T(t) - \mathbf{M}^{(\text{comp})T}(t)\big]\mathbf{y}^T(t) + \sum_{k=1}^N \text{diag}(\boldsymbol{\delta}_{\text{next}=k})\mathbf{z}_k^T(t) + \text{diag}(\boldsymbol{\delta}_{B_i = 1})\mathbf{p}^T(t), \\
    &\dot{\mathbf{a}}_i^T(t) = \big[\mathbf{Q}^T(t) - \mathbf{M}^{(i)T}(t)\big]\mathbf{a}_i^T(t) + \mathbf{y}^T(t) + \mathbf{p}^T(t).\label{eq:final_column_eq}
\end{align}

By extracting the matrix coefficients multiplying the state components themselves, we directly obtain the diagonal blocks of the overall system matrix $\mathbf{A}(t)$:
\begin{align}
    \mathbf{A}_{i, i}(t) &= \mathbf{Q}^T(t) - \mathbf{M}^{(i)T}(t), \quad i=1,\dots,N, \notag \\
    \mathbf{A}_{N+1, N+1}(t) &= \mathbf{Q}^T(t) - \mathbf{M}^{(\text{comp})T}(t), \notag \\
    \mathbf{A}_{N+1+i, N+1+i}(t) &= \mathbf{Q}^T(t) - \mathbf{M}^{(\text{next}=i)T}(t) - \lambda_i(t)\text{diag}(\boldsymbol{\delta}_{B_i=1}), \notag \\
    &\quad i=1,\dots,N, \notag \\
    \mathbf{A}_{2N+2, 2N+2}(t) &= \mathbf{Q}_{\text{red}}^T(t).
\end{align}

Note that the source term simplifies entirely to $\mathbf{b}(t) = [\boldsymbol{\beta}(t), \mathbf{0}, \dots, \mathbf{0}]^T$ coming from Eqn.~\eqref{eq:source_term}. By inspecting equations \eqref{eq:source_term}--\eqref{eq:final_column_eq}, we observe a sequential dependency among the states:

\begin{itemize}
    \item $\dot{\mathbf{p}}_{\text{red}}^T(t)$ depends exclusively on $\mathbf{p}_{\text{red}}^T(t)$.
    \item $\dot{\mathbf{z}}_i^T(t)$ depends on its own state $\mathbf{z}_i^T(t)$ and the underlying probabilities $\mathbf{p}_{\text{red}}^T(t)$.
    \item $\dot{\mathbf{y}}^T(t)$ depends on its own state $\mathbf{y}^T(t)$, the waiting age vectors $\mathbf{z}_k^T(t)$ for all classes $k=1,\dots,N$, and $\mathbf{p}_{\text{red}}^T(t)$.
    \item $\dot{\mathbf{a}}_i^T(t)$ depends on its own state $\mathbf{a}_i^T(t)$, the in-service age vector $\mathbf{y}^T(t)$, and $\mathbf{p}_{\text{red}}^T(t)$.
\end{itemize}
Because each vector derivative relies solely on its corresponding state and the states positioned structurally above it in the stacked column vector $\mathbf{x}(t)$, the entries strictly to the \textit{left} of the main diagonal matrix blocks of $\mathbf{A}$ are equal to zero. Hence, $\mathbf{A}$ is a block upper triangular matrix.
\end{proof}

Consider the nonhomogeneous age moments system given by $\dot{\mathbf{x}} = \mathbf{A}(t)\mathbf{x}(t)+\mathbf{b}(t)$. To analyze the evolution of the forced system, we define the fundamental matrix $\mathbf{\Phi}(t)$ associated with the homogeneous system, where $\mathbf{\Phi}(t)$ satisfies $\dot{\mathbf{\Phi}}(t) = \mathbf{A}(t) \mathbf{\Phi}(t)$ with $\mathbf{\Phi{(0)}} = \mathbf{I}$. Since $\mathbf{A}(t)$ is piecewise contimuous everywhere, a fundamental matrix $\mathbf{\Phi}(t)$ exists, is uniquely determined by its initial condition at $t=0$, and satisfies the matrix differential equation $\dot{\mathbf{\Phi}}(t) = \mathbf{A}(t) \mathbf{\Phi}(t)$ for all $t$ as established in \cite[Section 2.3]{antsaklis2006linear}. For a T-periodic system, the monodromy matrix is denoted as $\mathbf{\Phi}(T)$ i.e., the fundamental matrix evaluated over one full period. Its eigenvalues $\mu_i$ are called the Floquet multipliers of the system as shown in \cite{slane2011analysis}. To connect the structural properties of $\mathbf{A}(t)$ to the Floquet multipliers $\mu_i$
, we establish the following proposition.

\begin{proposition}[Block Upper Triangular Structure of the Fundamental Matrix] \label{prop_of_stm}
Let $\mathbf{A}(t)$ be a block upper triangular matrix with $n$ diagonal blocks 
$\mathbf{A}_{11}(t), \ldots, \mathbf{A}_{nn}(t)$. Then the state transition matrix 
$\mathbf{\Phi}(t,0)$ is also block upper triangular with the same block structure, 
and each diagonal block satisfies
\begin{equation}
    \dot{\mathbf{\Phi}}_{ii}(t) = \mathbf{A}_{ii}(t)\,\mathbf{\Phi}_{ii}(t), 
    \quad \mathbf{\Phi}_{ii}(0) = \mathbf{I}, 
    \quad i = 1, \ldots, n.
\end{equation}
\end{proposition}

\begin{proof}

We proceed by mathematical induction on the number of diagonal blocks, $n$, in the block upper triangular matrix $\mathbf{A}(t)$. Let the fundamental state transition matrix satisfy the homogeneous differential equation $\dot{\mathbf{\Phi}}(t,0) = \mathbf{A}(t)\mathbf{\Phi}(t,0)$ with the initial condition $\mathbf{\Phi}(0,0) = \mathbf{I}$.

\textbf{Base Case ($n=1$):} 
If $\mathbf{A}(t)$ has only a single block $\mathbf{A}_{11}(t)$, the matrix is trivially block upper triangular. The differential equation reduces to $\dot{\mathbf{\Phi}}_{11}(t) = \mathbf{A}_{11}(t)\mathbf{\Phi}_{11}(t)$ with $\mathbf{\Phi}_{11}(0) = \mathbf{I}$, which satisfies the proposition statement.

\textbf{Inductive Step:} 
Assume the proposition holds true for any block upper triangular matrix with $k$ diagonal blocks. We must show it holds for a matrix with $k+1$ blocks.

Let $\mathbf{A}(t)$ be a block upper triangular matrix with $k+1$ diagonal blocks. We partition $\mathbf{A}(t)$ into a $2 \times 2$ block structure:
\begin{equation}
    \mathbf{A}(t) = \begin{bmatrix} \mathbf{A}_{11}(t) & \mathbf{A}_{12}(t) \\ \mathbf{0} & \mathbf{A}_{22}(t) \end{bmatrix},
\end{equation}
where $\mathbf{A}_{11}(t)$ is a block upper triangular matrix containing the first $k$ diagonal blocks, and $\mathbf{A}_{22}(t)$ is the $(k+1)$-th diagonal block of the original matrix. 

We partition the fundamental matrix $\mathbf{\Phi}(t,0)$ in the same way:
\begin{equation}
    \mathbf{\Phi}(t,0) = \begin{bmatrix} \mathbf{\Phi}_{11}(t) & \mathbf{\Phi}_{12}(t) \\ \mathbf{\Phi}_{21}(t) & \mathbf{\Phi}_{22}(t) \end{bmatrix}.
\end{equation}

Substituting these partitions into the differential equation yields:
\begin{align}
    &\begin{bmatrix} \dot{\mathbf{\Phi}}_{11}(t) & \dot{\mathbf{\Phi}}_{12}(t) \\ \dot{\mathbf{\Phi}}_{21}(t) & \dot{\mathbf{\Phi}}_{22}(t) \end{bmatrix} = \begin{bmatrix} \mathbf{A}_{11}(t) & \mathbf{A}_{12}(t) \\ \mathbf{0} & \mathbf{A}_{22}(t) \end{bmatrix} \begin{bmatrix} \mathbf{\Phi}_{11}(t) & \mathbf{\Phi}_{12}(t) \\ \mathbf{\Phi}_{21}(t) & \mathbf{\Phi}_{22}(t) \end{bmatrix} \notag \\[6pt]
    &= \begin{bmatrix} \mathbf{A}_{11}(t)\mathbf{\Phi}_{11}(t) + \mathbf{A}_{12}(t)\mathbf{\Phi}_{21}(t) & \mathbf{A}_{11}(t)\mathbf{\Phi}_{12}(t) + \mathbf{A}_{12}(t)\mathbf{\Phi}_{22}(t) \\ \mathbf{A}_{22}(t)\mathbf{\Phi}_{21}(t) & \mathbf{A}_{22}(t)\mathbf{\Phi}_{22}(t) \end{bmatrix}.
\end{align}

Equating the bottom-left blocks, we obtain the initial value problem:
\begin{equation}
    \dot{\mathbf{\Phi}}_{21}(t) = \mathbf{A}_{22}(t)\mathbf{\Phi}_{21}(t).
\end{equation}
Because the initial condition of the full system is $\mathbf{\Phi}(0,0) = \mathbf{I}$, the off-diagonal block must satisfy $\mathbf{\Phi}_{21}(0) = \mathbf{0}$. By the existence and uniqueness theorem for linear ordinary differential equations, the only solution to this homogeneous initial value problem is the trivial solution, $\mathbf{\Phi}_{21}(t) = \mathbf{0}$ for all $t$. 

Consequently, $\mathbf{\Phi}(t,0)$ maintains a block upper triangular structure at this partition level. Substituting $\mathbf{\Phi}_{21}(t) = \mathbf{0}$ back into the diagonal block equations gives:
\begin{align}
    \dot{\mathbf{\Phi}}_{11}(t) &= \mathbf{A}_{11}(t)\mathbf{\Phi}_{11}(t), \quad \mathbf{\Phi}_{11}(0) = \mathbf{I}, \label{eq:phi_11} \\
    \dot{\mathbf{\Phi}}_{22}(t) &= \mathbf{A}_{22}(t)\mathbf{\Phi}_{22}(t), \quad \mathbf{\Phi}_{22}(0) = \mathbf{I}. \label{eq:phi_22}
\end{align}

Equation \eqref{eq:phi_22} shows that the $(k+1)$-th diagonal block decouples and satisfies the required condition. For the top-left block, since $\mathbf{A}_{11}(t)$ is a block upper triangular matrix with $k$ blocks, the inductive hypothesis guarantees that the solution $\mathbf{\Phi}_{11}(t)$ to equation \eqref{eq:phi_11} is also block upper triangular, and that all $k$ of its diagonal blocks satisfy $\dot{\mathbf{\Phi}}_{ii}(t) = \mathbf{A}_{ii}(t)\mathbf{\Phi}_{ii}(t)$.

Therefore, by mathematical induction, the fundamental matrix $\mathbf{\Phi}(t,0)$ is block upper triangular with $n$ blocks, and every diagonal block satisfies $\dot{\mathbf{\Phi}}_{ii}(t) = \mathbf{A}_{ii}(t)\mathbf{\Phi}_{ii}(t)$ for $i=1,\dots,n$.
\end{proof}

The following lemma establishes the key properties of $\mathbf{A}(t)$.

\begin{lemma}[Properties of $\mathbf{A}(t)$] \label{prop_of_A}
The diagonal blocks $\mathbf{A}_{ii}(t)$ defined in 
equations~(29)--(32) satisfy:
\begin{enumerate}
    \item[(i)] $[\mathbf{A}_{ii}(t)]_{jk} \geq 0$ for all $j \neq k$,
    \item[(ii)] $\sum_{k} [\mathbf{A}_{ii}(t)]_{kj} < 0$ for all $j$.
\end{enumerate}
\end{lemma}
\begin{proof}

We address the two properties separately, relying on the structures of the diagonal blocks $\mathbf{A}_{ii}(t)$ provided in Corollary 1.

\textbf{Part (i): Non-negativity of off-diagonal entries} \\
Recall that $\mathbf{Q}(t)$ is the generator matrix of a Continuous-Time Markov Chain (CTMC). By definition, its off-diagonal entries represent transition rates and are therefore non-negative: $Q_{jk}(t) \ge 0$ for all $j \neq k$. Transposing the matrix preserves this property, meaning $[\mathbf{Q}^T(t)]_{jk} \ge 0$ for $j \neq k$.

The sparse matrices $\mathbf{M}^{(i)}(t)$, $\mathbf{M}^{(\text{comp})}(t)$, and $\mathbf{M}^{(\text{next}=i)}(t)$ represent the transition rates for specific state changes. Because these matrices only contain the rates for these specific events from the generator matrix $\mathbf{Q}(t)$, the matrix difference $\mathbf{Q}(t) - \mathbf{M}^{(\cdot)}(t)$ effectively just sets those specific entries to zero. Since the original off-diagonal entries of $\mathbf{Q}(t)$ are non-negative, setting a subset of them to zero guarantees that all remaining off-diagonal entries stay non-negative. Similarly, the term $\lambda_i(t)\text{diag}(\boldsymbol{\delta}_{B_i=1})$ is explicitly a diagonal matrix, meaning it only modifies the main diagonal and leaves the off-diagonal entries unchanged. Hence, the off-diagonal entries of the resulting blocks $\mathbf{A}_{ii}(t)$ remain completely unaffected. Thus, $[\mathbf{A}_{ii}(t)]_{jk} \ge 0$ for all $j \neq k$.

\textbf{Part (ii): Strictly negative row sums} \\
The row sums of the generator matrix are zero: $\mathbf{Q}(t)\mathbf{1} = \mathbf{0}$. Since the matrices $\mathbf{M}^{(\cdot)}(t)$ have zero on their diagonals and non-negative off-diagonal transition rates, the row sum for any of these sparse matrices is strictly positive. This yields $\mathbf{M}^{(\cdot)}(t)\mathbf{1} > \mathbf{0}$. Furthermore, the indicator terms like $\lambda_i(t)\text{diag}(\boldsymbol{\delta}_{B_i=1})$ are diagonal matrices with non-negative entries.

We evaluate the column sum for each block type by multiplying them with $\mathbf{1}^T$. For the first set of diagonal blocks, $\mathbf{A}_{i,i}(t)$:
\begin{align*}
    \mathbf{1}^T \mathbf{A}_{i,i}(t) &= \mathbf{1}^T \big(\mathbf{Q}^T(t) - \mathbf{M}^{(i)T}(t)\big) \notag \\
    &= \big(\mathbf{Q}(t)\mathbf{1}\big)^T - \big(\mathbf{M}^{(i)}(t)\mathbf{1}\big)^T.
\end{align*}
Since $\mathbf{Q}(t)\mathbf{1} = \mathbf{0}$ and $\mathbf{M}^{(i)}(t)\mathbf{1} > \mathbf{0}$, this directly results in $\mathbf{1}^T \mathbf{A}_{i,i}(t) < \mathbf{0}^T$.

For the block $\mathbf{A}_{N+1,N+1}(t)$:
\begin{align*}
    \mathbf{1}^T \mathbf{A}_{N+1,N+1}(t) &= \mathbf{1}^T \big(\mathbf{Q}^T(t) - \mathbf{M}^{(\text{comp})T}(t)\big) \notag \\
    &= \big(\mathbf{Q}(t)\mathbf{1}\big)^T - \big(\mathbf{M}^{(\text{comp})}(t)\mathbf{1}\big)^T.
\end{align*}
Since $\mathbf{M}^{(\text{comp})}(t)\mathbf{1} > \mathbf{0}$, this yields $\mathbf{1}^T \mathbf{A}_{N+1,N+1}(t) < \mathbf{0}^T$.

For the next set of diagonal blocks, $\mathbf{A}_{N+1+i,N+1+i}(t)$:
\begin{align*}
    \mathbf{1}^T \mathbf{A}_{N+1+i,N+1+i}(t) &= \mathbf{1}^T \big(\mathbf{Q}^T(t) - \mathbf{M}^{(\text{next}=i)T}(t) - \lambda_i(t)\text{diag}(\boldsymbol{\delta}_{B_i=1})\big) \notag \\
    &= \big(\mathbf{Q}(t)\mathbf{1}\big)^T - \big(\mathbf{M}^{(\text{next}=i)}(t)\mathbf{1}\big)^T - \mathbf{1}^T \big(\lambda_i(t)\text{diag}(\boldsymbol{\delta}_{B_i=1})\big).
\end{align*}
Since $\mathbf{Q}(t)\mathbf{1} = \mathbf{0}$, $\mathbf{M}^{(\text{next}=i)}(t)\mathbf{1} > \mathbf{0}$, and the product involving the non-negative diagonal indicator matrix yields $\mathbf{1}^T \big(\lambda_i(t)\text{diag}(\boldsymbol{\delta}_{B_i=1})\big) \ge \mathbf{0}^T$, subtracting these terms guarantees the entire result is strictly negative: $\mathbf{1}^T \mathbf{A}_{N+1+i,N+1+i}(t) < \mathbf{0}^T$. 

To prove the property for the last diagonal block, $\mathbf{A}_{i,i}(t) = \mathbf{Q}^{T}_{red}(t)$, we define $\mathbf{Q}(t)$, the generator matrix of the CTMC, in block form as:
\begin{equation*}
    \mathbf{Q}(t) = \begin{bmatrix} \mathbf{Q}_{11}(t) & \mathbf{q}_{1N}(t) \\ \mathbf{q}_{N1}(t) & q_{NN}(t) \end{bmatrix},
\end{equation*}
where $\mathbf{Q}_{11}(t) \in \mathbb{R}^{(N-1) \times (N-1)}$, $\mathbf{q}_{1N}(t) = \mathbf{q}_{N1}^\top(t) \in \mathbb{R}^{(N-1) \times 1}$, and $q_{NN}(t) \in \mathbb{R}$. Recall the definitions of $\mathbf{T} = \begin{bmatrix} \mathbf{I}_{N-1} & -\mathbf{1} \end{bmatrix}$ and $\mathbf{E} = \begin{bmatrix} \mathbf{I}_{N-1} \\ \mathbf{0}^T \end{bmatrix}$ as introduced in the proof of Proposition \ref{prop_reduced}. First, we derive the expression for $\mathbf{Q}_{\text{red}}(t)$:
\begin{align*}
    \mathbf{Q}_{\text{red}}(t) &= \mathbf{T} \mathbf{Q}(t) \mathbf{E} \notag \\
    &= \mathbf{T} \begin{bmatrix} \mathbf{Q}_{11}(t) & \mathbf{q}_{1N}(t) \\ \mathbf{q}_{N1}(t) & q_{NN}(t) \end{bmatrix} \begin{bmatrix} \mathbf{I}_{N-1} \\ \mathbf{0}^T \end{bmatrix} \notag \\
    &= \mathbf{T} \begin{bmatrix} \mathbf{Q}_{11}(t) \\ \mathbf{q}_{N1}(t) \end{bmatrix} \notag \\
    &= \begin{bmatrix} \mathbf{I}_{N-1} & -\mathbf{1} \end{bmatrix} \begin{bmatrix} \mathbf{Q}_{11}(t) \\ \mathbf{q}_{N1}(t) \end{bmatrix} \notag \\
    &= \mathbf{Q}_{11}(t) - \mathbf{1} \mathbf{q}_{N1}(t).
\end{align*}

Next, we evaluate the row sums of the reduced matrix by post-multiplying by $\mathbf{1}$:
\begin{equation*}
    \mathbf{Q}_{\text{red}}(t) \mathbf{1} = \left[ \mathbf{Q}_{11}(t) - \mathbf{1} \mathbf{q}_{N1}(t) \right] \mathbf{1} = \mathbf{Q}_{11}(t) \mathbf{1} - \mathbf{1} \big( \mathbf{q}_{N1}(t) \mathbf{1} \big).
\end{equation*}

Because $\mathbf{Q}(t)$ is a valid CTMC generator matrix, its row sums must equal zero ($\mathbf{Q}(t) \mathbf{1} = \mathbf{0}$). This gives us two key observations:
\begin{enumerate}
    \item From the first $N-1$ rows: $\mathbf{Q}_{11}(t) \mathbf{1} + \mathbf{q}_{1N}(t) = \mathbf{0} \implies \mathbf{Q}_{11}(t) \mathbf{1} \leq \mathbf{0}$ (since off-diagonal rates $\mathbf{q}_{1N}(t) \geq \mathbf{0}$).
    \item From the last row: $\mathbf{q}_{N1}(t) \mathbf{1} + q_{NN}(t) = 0 \implies \mathbf{q}_{N1}(t) \mathbf{1} = -q_{NN}(t)$.
\end{enumerate}

Note that since an arrival can always change the state when the system is empty, and a service completion can always change the state when the system is busy, the CTMC has no absorbing states. Hence, the rate of leaving the state $s$ indexed by $\sigma(s) = N$ is strictly positive, meaning $q_{NN}(t) < 0$. Therefore:
\begin{equation*}
    \mathbf{q}_{N1}(t) \mathbf{1} = -q_{NN}(t) > 0.
\end{equation*}

Substituting these results back into our row sum equation yields:
\begin{equation*}
    \mathbf{Q}_{\text{red}}(t) \mathbf{1} = \underbrace{\mathbf{Q}_{11}(t) \mathbf{1}}_{\leq \mathbf{0}} - \mathbf{1} \underbrace{\big( \mathbf{q}_{N1}(t) \mathbf{1} \big)}_{> 0} < \mathbf{0}.
\end{equation*}

Thus, all row sums of the reduced matrix $\mathbf{Q}_{red}(t)$ are strictly negative. Thus, $\mathbf{Q}_{\text{red}}(t)\mathbf{1} < \mathbf{0}$, which implies:
\begin{equation*}
    \mathbf{1}^T \mathbf{A}_{2N+2, 2N+2}(t) = \mathbf{1}^T \mathbf{Q}_{\text{red}}^T(t) = \big(\mathbf{Q}_{\text{red}}(t)\mathbf{1}\big)^T < \mathbf{0}^T.
\end{equation*}

Therefore, the column sums of every block are strictly negative.

\end{proof}

The following lemma shows that the structural properties of $\mathbf{A}(t)$ are inherited by the fundamental matrix $\mathbf{\Phi}(t)$.

\begin{lemma}[Properties of $\boldsymbol{\Phi}(t)$] \label{prop_of_Phi}

Let $\mathbf{A}(t)$ be a matrix satisfying the properties in Lemma \ref{prop_of_A}.
Then, the diagonal blocks of the fundamental matrix $\boldsymbol{\Phi}_{ii}(t)$ satisfy:

\begin{enumerate}
    \item[(i)] $[\mathbf{\Phi}_{ii}(t)]_{jk} \geq 0$ for all $j, k$ and $t \geq 0$,
    \item[(ii)] $\sum_{k} [\mathbf{\Phi}_{ii}(t)]_{kj} < 1$ for all $j$ and $t \geq 0$.
\end{enumerate}
\end{lemma}

\begin{proof}

\textbf{Part (i): Non-negativity of off-diagonal entries} \\
We first prove the non-negativity of $\mathbf{\Phi}_{ii}(t)$. Consider the solution $\mathbf{x}(t)$ to the differential equation $\dot{\mathbf{x}}(t) = \mathbf{A}_{ii}(t)\mathbf{x}(t)$ with the initial condition $\mathbf{x}(0) = \mathbf{e}_n$, where $\mathbf{e}_n$ is the standard basis column vector (all elements are $0$ except the $n$-th element, which is $1$). By definition, $\mathbf{x}(t)$ corresponds to the $n$-th row of $\mathbf{\Phi}_{ii}(t)$.

To show that the elements of $\mathbf{x}(t)$ never become negative, suppose that at some given time $t_0 \ge 0$, an element reaches zero such that $x_k(t_0) = 0$, while all other elements remain non-negative ($x_j(t_0) \ge 0$ for $j \ne k$). The time derivative of the $k$-th component at $t_0$ is:
\begin{equation*}
    \dot{x}_k(t_0) = \sum_{j} x_j(t_0) [\mathbf{A}_{ii}]_{jk}(t_0) = x_k(t_0) [\mathbf{A}_{ii}(t_0)]_{kk} + \sum_{j \ne k} x_j(t_0) [\mathbf{A}_{ii}(t_0)]_{jk}.
\end{equation*}
Since $x_k(t_0) = 0$, the first term vanishes, leaving:
\begin{equation*}
    \dot{x}_k(t_0) = \sum_{j \ne k} x_j(t_0) [\mathbf{A}_{ii}(t_0)]_{jk}.
\end{equation*}
By assumption, the off-diagonal elements $[\mathbf{A}_{ii}(t_0)]_{jk} \ge 0$, and $x_j(t_0) \ge 0$. Therefore, $\dot{x}_k(t_0) \ge 0$. This implies that whenever an element drops to zero, its rate of change is non-negative, preventing it from ever becoming negative. Hence, $\mathbf{x}(t) \ge \mathbf{0}^T$ for all $t \ge 0$. Because this property holds for any $n$-th row of $\mathbf{\Phi}_{ii}(t)$, we conclude that the entries of the matrix remain non-negative:
\begin{equation*}
    [\mathbf{\Phi}_{ii}(t)]_{jk} \ge 0 \quad \text{for all } j, k \text{ and } t \ge 0.
\end{equation*}
Thus, evaluated at time $T$, $[\mathbf{\Phi}_{ii}(T)]_{jk} \ge 0$.

\textbf{Part (ii): Row sums are strictly less than 1} \\
Next, we evaluate the row sums of $\mathbf{\Phi}_{ii}(t)$.Multiplying the differential equation $\mathbf{\dot{\Phi}}_{ii}(t) = \mathbf{\Phi}_{ii}(t)\mathbf{A}_{ii}(t)$ by the all-ones column vector $\mathbf{1}$, we obtain:
\begin{equation*}
    \mathbf{\dot{\Phi}}_{ii}(t) = \mathbf{\Phi}_{ii}(t)(\mathbf{A}_{ii}(t)\mathbf{1}\big).
\end{equation*}
Let $\mathbf{s}(t) = \mathbf{\Phi}_{ii}(t)\mathbf{1}$ denote the vector of row sums of $\mathbf{\Phi}_{ii}(t)$. By the given conditions, the row sums of $\mathbf{A}_{ii}(t)$ are strictly negative, so $\mathbf{A}_{ii}(t)\mathbf{1} < \mathbf{0}$. The derivative of the $j$-th row sum is:
\begin{equation*}
    \dot{\mathbf{s}}_j(t) = \big( \dot{\mathbf{\Phi}_{ii}}(t)\mathbf{1} \big)_j = \sum_{k} [\mathbf{\Phi}_{ii}]_{jk}(t) \big(\mathbf{A}_{ii}(t)\mathbf{1}\big)_k.
\end{equation*}
From Part 1, we know $[\mathbf{\Phi}_{ii}(t)]_{jk} \ge 0$. Because $\mathbf{\Phi}_{ii}(t)$ is a fundamental matrix solution for an ODE, it is invertible for all $t$, meaning no row can consist entirely of zeros. Combined with the non-negativity result from Part (i) ($[\mathbf{\Phi}_{ii}(t)]_{jk} \ge 0$), it follows that for every row $j$, there exists at least one index $k$ such that $[\mathbf{\Phi}_{ii}(t)]_{jk}(t) > 0$. Since every $\big(\mathbf{A}_{ii}(t)\mathbf{1}\big)_j < 0$, the linear combination must be strictly negative:
\begin{equation*}
    \dot{\mathbf{s}}_j(t) < 0 \quad \text{for all } j.
\end{equation*}
This demonstrates that every row sum $\mathbf{s}_j(t)$ strictly decreases for $t \ge 0$. At $t = 0$, $\mathbf{\Phi}_{ii}(0) = I$, making the initial row sums exactly $\mathbf{s}_j(0) = 1$. Because the row sums can only decrease starting from $1$, it follows that for any $t > 0$:
\begin{equation*}
    \mathbf{s}_j(t) = \big( \mathbf{\Phi}_{ii}(t)\mathbf{1} \big)_j < 1 \quad \text{for all } j.
\end{equation*}

\end{proof}

Together, these lemmas establish that $\mathbf{A}(t)$ and its fundamental matrix $\mathbf{\Phi}(t)$ share the same structural properties. We now characterize the class of matrices satisfying these properties.

\begin{corollary}[Eigenvalue Bound] \label{eigen_corr}
Let $\boldsymbol{\Phi}(t)$ be the fundamental matrix satisfying the properties established in Lemma~2. Then all eigenvalues of $\boldsymbol{\Phi}(t)$ are strictly less than one in absolute value for all time $t \geq 0$.

\end{corollary}

\begin{proof}

Let $\lambda(t)$ be an eigenvalue of a diagonal block $\mathbf{\Phi}_{ii}(t)$ of the fundamental matrix $\mathbf{\Phi}(t)$, and let
\[
\mathbf{v}(t) = [v_1(t), v_2(t), \dots, v_n(t)]^T
\]
be its corresponding eigenvector such that $\mathbf{v}(t) \neq \mathbf{0}$. By definition,
\begin{align*}
    \mathbf{\Phi}_{ii}(t) \mathbf{v}(t) = \lambda(t) \mathbf{v}(t) \implies \lambda(t) v_j(t) &= \sum_{k=1}^n [\mathbf{\Phi}_{ii}(t)]_{jk}v_k(t), \notag \; \forall j \in \{1, \dots, n\}.
\end{align*}
Taking the absolute value and applying the triangle inequality, we have
\begin{equation*}
    |\lambda(t)| |v_j(t)| = \left| \sum_{k=1}^n [\mathbf{\Phi}_{ii}(t)]_{jk} v_k(t) \right| \leq \sum_{k=1}^n |[\mathbf{\Phi}_{ii}(t)]_{jk}|v_k(t)|.
\end{equation*}
Since $\mathbf{\Phi}_{ii}(t)]_{jk} \geq 0$ according to the properties established in Lemma \ref{prop_of_Phi}, it follows that
\begin{equation*}
    |\lambda(t)| |v_j(t)| \leq \sum_{k=1}^n \mathbf{\Phi}_{ii}(t)]_{jk} |v_k(t)|.
\end{equation*}
Summing over all $j$,
\begin{equation*}
    |\lambda(t)| \sum_{j=1}^n |v_j(t)| \leq \sum_{j=1}^n \sum_{k=1}^n [\mathbf{\Phi}_{ii}(t)]_{jk} |v_k(t)|.
\end{equation*}
Interchanging the order of summation on the right-hand side yields
\begin{equation*}
    |\lambda(t)| \sum_{j=1}^n |v_j(t)| \leq \sum_{k=1}^n |v_k(t)| \left( \sum_{j=1}^n [\mathbf{\Phi}_{ii}(t)]_{jk} \right).
\end{equation*}
Let $c(t) = \max_k \sum_{j=1}^n [\mathbf{\Phi}_{ii}(t)]_{jk}$ denote the maximum column sum of $\mathbf{\Phi}_{ii}(t)$. Thus, we can write
\begin{equation*}
    |\lambda(t)| \sum_{j=1}^n |v_j(t)| \leq c(t) \sum_{k=1}^n |v_k(t)|.
\end{equation*}
From the properties in Lemma \ref{prop_of_Phi}, we have $c(t) < 1$. Since $\sum_{j=1}^n |v_j(t)| > 0$, we conclude that
\begin{equation*}
    |\lambda(t)| \leq c(t) < 1.
\end{equation*}
Consequently, the eigenvalues of each diagonal block of $\mathbf{\Phi}(t)$ are strictly less than one in magnitude. As established in Proposition \ref{prop_of_stm}, $\mathbf{\Phi}(t)$ is a block upper-triangular matrix. Because the eigenvalues of a block upper-triangular matrix consist precisely of the union of the eigenvalues of its diagonal blocks \cite[Proposition~5.5.13]{bernstein2009matrix}, it immediately follows that all eigenvalues of $\mathbf{\Phi}(t)$ strictly satisfy $|\lambda_i(\mathbf{\Phi}(t))| < 1$.
\end{proof}

We now establish the key property of the monodromy matrix 
$\boldsymbol{\Phi}(T)$, which governs the steady-state behavior of the system.

\begin{corollary}[Floquet Multiplier Bound]
The Floquet multipliers of the system \eqref{eq_compact}, i.e., the 
eigenvalues of the monodromy matrix $\boldsymbol{\Phi}(T)$, are strictly less than one in absolute value.
\end{corollary}

\begin{proof}
Follows directly from Corollary~2 applied at $t = T$.
\end{proof}

For the convenience of the reader, we restate the main convergence theorem here before proceeding to the proof.

\setcounter{theorem}{1}
\begin{theorem}[Existence and Uniqueness of the Periodic Steady State]
\label{existence_thm}
The system of linear ODEs established in Theorem~\ref{theorem_main} admits a unique periodic steady state, denoted by $\mathbf{x}^\star(t)$, satisfying $\mathbf{x}^\star(t+T) = \mathbf{x}^\star(t)$ for all $t \geq 0$. Moreover, for any initial condition $\mathbf{x}(0)$, the solution converges exponentially to $\mathbf{x}^\star(t)$ in the standard Euclidean norm, satisfying:
\begin{equation}
    \|\mathbf{x}(t) - \mathbf{x}^\star(t)\|_2 \leq m e^{-\gamma t} \|\mathbf{x}(0) - \mathbf{x}^\star(0)\|_2,
\end{equation}
for some constants $\gamma > 0$ and $m \geq 1$.
\end{theorem}

\begin{proof}
Throughout this proof, $\|\cdot\|_2$ denotes the standard Euclidean vector norm or its induced matrix norm. For a positive-definite Hermitian matrix $\mathbf{H} \in \mathbb{C}^{n \times n}$, the weighted vector norm is denoted by
\begin{equation*}
    \|\mathbf{x}\|_\mathbf{H} = \sqrt{\mathbf{x}^H\mathbf{H}\mathbf{x}},
\end{equation*}
and its induced matrix spectral norm is defined as
\begin{equation*}
    \|\mathbf{A}\|_\mathbf{H} \triangleq \max_{\mathbf{x}\neq\mathbf{0}} \frac{\|\mathbf{A}\mathbf{x}\|_\mathbf{H}}{\|\mathbf{x}\|_\mathbf{H}}.
\end{equation*}
The corresponding weighted logarithmic matrix norm is defined as
\begin{equation*}
    \mu_{\mathbf{H}}[\mathbf{A}] \triangleq \lim_{\Delta \to 0^+} \frac{\|\mathbf{I} + \Delta \mathbf{A}\|_\mathbf{H} - 1}{\Delta}.
\end{equation*}

By the extended Floquet theory for inhomogeneous systems \cite[Theorem 4.1]{slane2011analysis}, the solution after $n$ periods splits into an unforced transient response and an accumulated forcing contribution:
\begin{equation*}
\mathbf{x}(nT) = \mathbf{\Phi}(T)^n \mathbf{x}(0) + \sum_{k=1}^{n} \mathbf{\Phi}(T)^k \Lambda,
\end{equation*}
where $\Lambda = \int_0^T \mathbf{\Phi}(\tau)^{-1} \mathbf{g}(\tau) d\tau$. Because the Floquet multipliers (the eigenvalues of the monodromy matrix $\mathbf{\Phi}(T)$) satisfy $\max_{i} |\lambda_i(\mathbf{\Phi}(T))| < 1$, the transient term $\mathbf{\Phi}(T)^n \mathbf{x}(0)$ decays to zero and the Neumann series $\sum_{k=1}^{n} \mathbf{\Phi}(T)^k $converges as $n \to \infty$ \cite[Theorem 4.4]{slane2011analysis}. This establishes the existence of a unique, $T$-periodic steady-state solution $\mathbf{x}^\star(t)$. To analyze the stability of this steady state, we evaluate the standard Euclidean norm of the trajectory deviation, $\|\mathbf{x}(t) - \mathbf{x}^\star(t)\|_2$. By the Floquet–Lyapunov theorem \cite[Theorem 3.1]{slane2011analysis}, the fundamental matrix can be factored as $\mathbf{\Phi}(t) = \mathbf{F}(t)e^{\mathbf{K}t}$, where $\mathbf{F}(t)$ is a nonsingular, continuous, $T$-periodic matrix and $\mathbf{K}$ is a constant matrix. Note that since $\mathbf{\Phi}(0)=\mathbf{I}$, 
\begin{align}
    \mathbf{\Phi}(0) = \mathbf{F}(0) e^{\mathbf{K} \cdot 0} \notag = \mathbf{F}(0) = \mathbf{I} = \mathbf{F}(T)
\end{align} 
Applying the submultiplicative property of induced matrix spectral norms in (\ref{eq:deviation}c) to the deviation yields:
\begin{subequations}
\refstepcounter{equation}
\begin{align} \label{eq:deviation}
\|\mathbf{x}(t) - \mathbf{x}^\star(t)\|_2 &\overset{\mathclap{(\alph{equation})}}{=} \|\mathbf{\Phi}(t) \big(\mathbf{x}(0) - \mathbf{x}^\star(0)\big)\|_2 \tag{\theparentequation} \\\noalign{\refstepcounter{equation}}
&\overset{\mathclap{(\alph{equation})}}{\leq} \|\mathbf{F}(t)e^{\mathbf{K}t}(\mathbf{x}(0) - \mathbf{x}^\star(0))\|_2 \notag \\ \noalign{\refstepcounter{equation}}
&\overset{\mathclap{(\alph{equation})}}{\leq} \|\mathbf{F}(t)\|_2 \|e^{\mathbf{K}t}\|_2 \|\mathbf{x}(0) - \mathbf{x}^\star(0)\|_2. \notag \\
\noalign{\refstepcounter{equation}}
&\overset{\mathclap{(\alph{equation})}}{\leq} M \|e^{\mathbf{K}t}\|_2 \|\mathbf{x}(0) - \mathbf{x}^\star(0)\|_2, \notag
\end{align}
\end{subequations}
where $M \triangleq \max_{t \in [0,T]} \|\mathbf{F}(t)\|_2 < \infty$. Because $\mathbf{F}(t)$ is continuous and $T$-periodic, its norm $\|\mathbf{F}(t)\|_2$ is a continuous scalar function evaluated over the compact interval $[0, T]$. By the Extreme Value Theorem, this maximum is guaranteed to be well-defined and finite.
To bound the remaining continuous-time component $\|e^{\mathbf{K}t}\|_2$, we utilize the weighted logarithmic matrix norm. For a positive-definite Hermitian matrix $\mathbf{H}$ with a unique square root $\mathbf{H}^{1/2}$, the induced weighted norm relates to the standard matrix norm via a change of variables, obtained by setting $\mathbf{z} = \mathbf{H}^{1/2}\mathbf{x}$:
\begin{align}
\|\mathbf{A}\|_\mathbf{H} &= \max_{\mathbf{x} \neq \mathbf{0}} \frac{\|\mathbf{A}\mathbf{x}\|_\mathbf{H}}{\|\mathbf{x}\|_\mathbf{H}} \notag \\ &= \max_{\mathbf{x} \neq \mathbf{0}} \frac{\|\mathbf{H}^{1/2}\mathbf{A}\mathbf{x}\|_2}{\|\mathbf{H}^{1/2}\mathbf{x}\|_2} \notag \\
&= \max_{\mathbf{z} \neq \mathbf{0}} \frac{\|\mathbf{H}^{1/2}\mathbf{A}\mathbf{H}^{-1/2}\mathbf{z}\|_2}{\|\mathbf{z}\|_2} \notag \\
&= \|\mathbf{H}^{1/2}\mathbf{A}\mathbf{H}^{-1/2}\|_2.\end{align}
Since $\mathbf{H}$ is positive-definite Hermitian, the spectral norms of its square-root components reduce directly to their extreme eigenvalues, namely $\|\mathbf{H}^{1/2}\|_2 = \sqrt{\lambda_{\max}(\mathbf{H})}$ and $\|\mathbf{H}^{-1/2}\|_2 = 1/\sqrt{\lambda_{\min}(\mathbf{H})}$. By expanding the matrix exponential via the identity $e^{\mathbf{K}t} = \mathbf{H}^{-1/2}(\mathbf{H}^{1/2}e^{\mathbf{K}t}\mathbf{H}^{-1/2})\mathbf{H}^{1/2}$ and applying the submultiplicative property of the induced matrix spectral norm, we can directly separate the matrix terms in (\ref{eq:similar}a),
\begin{subequations}
\refstepcounter{equation}
\begin{align} \label{eq:similar}
\|e^{\mathbf{K}t}\|_2  &\overset{\mathclap{(\alph{equation})}}{\leq} \|\mathbf{H}^{-1/2}\|_2 \cdot \|\mathbf{H}^{1/2}e^{\mathbf{K}t}\mathbf{H}^{-1/2}\|_2 \cdot \|\mathbf{H}^{1/2}\|_2 \tag{\theparentequation} \\ \noalign{\refstepcounter{equation}}
&\overset{\mathclap{(\alph{equation})}}{=}  \frac{1}{\sqrt{\lambda_{\min}(\mathbf{H})}} \cdot \|e^{\mathbf{K}t}\|_\mathbf{H} \cdot \sqrt{\lambda_{\max}(\mathbf{H})} \notag \\ \noalign{\refstepcounter{equation}}
&\overset{\mathclap{(\alph{equation})}}{=}  \sqrt{\frac{\lambda_{\max}(\mathbf{H})}{\lambda_{\min}(\mathbf{H})}} \|e^{\mathbf{K}t}\|_\mathbf{H}. \notag
\end{align}
\end{subequations}
Invoking \cite[Lemma 1.3]{hu2012bounds}, the weighted matrix exponential is strictly bounded by its corresponding weighted logarithmic norm, such that $\|e^{\mathbf{K}t}\|_\mathbf{H} \leq \exp(\mu_\mathbf{H}[\mathbf{K}]t)$. Rather than selecting an arbitrary matrix, we specifically construct $\mathbf{H}$ to tightly bound the exponential decay of the system. Leveraging the fact that $\mathbf{F}(T) = \mathbf{I}$, the principal relation reduces to $\mathbf{\Phi}(T) = e^{\mathbf{K}T}$. This directly maps the magnitudes of the eigenvalues as
\begin{align}
|\lambda(\mathbf{\Phi}(T))| = |e^{\lambda(\mathbf{K})T}| = e^{\text{Re}(\lambda(\mathbf{K}))T}.
\end{align}
Then, bounding the Floquet multipliers strictly below unity in absolute value ($\max_{i} |\lambda_i(\mathbf{\Phi}(T))| < 1$) via Corollary 3 immediately implies that $\max_i \{\text{Re}(\lambda_i(\mathbf{K}))\} < 0$. Following  \cite[Theorem 2.8]{hu2012bounds}, for any real scalar $\sigma > \max_{i} \{\text{Re}(\lambda_i(\mathbf{K}))\}$, there exists a unique positive-definite Hermitian matrix $\mathbf{H}$ that fulfills the Lyapunov equation
\begin{equation}
    (\mathbf{K} - \sigma \mathbf{I})^*\mathbf{H} + \mathbf{H}(\mathbf{K} - \sigma \mathbf{I}) = -2\mathbf{I}.
\end{equation}
Crucially, evaluating the system under this specific matrix yields a weighted logarithmic norm for $\mathbf{K}$ that strictly satisfies:
\begin{equation}
    \mu_\mathbf{H}[\mathbf{K}] = \sigma - \frac{1}{\lambda_{\max}(\mathbf{H})}.
\end{equation}
By choosing $\sigma = \max_{i} \{ \text{Re}(\lambda_i(\mathbf{K})) \} + \epsilon$ for an arbitrarily small $\epsilon > 0$, we ensure the logarithmic norm is strictly negative. Note that $\lambda_{max}(\mathbf{H}) > 0$ since $\mathbf{H}$ is positive-definite. We define the constant decay rate $-\gamma \triangleq \mu_\mathbf{H}[\mathbf{K}] < 0$. Substituting this continuous decay term back through the norm inequalities produces the final global exponential bound:
\begin{align}
\|\mathbf{x}(t) - \mathbf{x}^\star(t)\|_2 &\leq M \sqrt{\frac{\lambda_{\max}(\mathbf{H})}{\lambda_{\min}(\mathbf{H})}} \exp(\mu_\mathbf{H}[\mathbf{K}]t) \|\mathbf{x}(0) - \mathbf{x}^\star(0)\|_2 \notag \\
&= m e^{-\gamma t} \|\mathbf{x}(0) - \mathbf{x}^\star(0)\|_2,
\end{align}
where the scaling constant is defined as $m \triangleq M \sqrt{\frac{\lambda_{\max}(\mathbf{H})}{\lambda_{\min}(\mathbf{H})}}$. 
Furthermore, $M = \max_{t \in [0,T]} \|\mathbf{F}(t)\|_2 \geq \|\mathbf{F}(0)\|_2 = \|\mathbf{I}\|_2 = 1$ and $\frac{\lambda_{\max}(\mathbf{H})}{\lambda_{\min}(\mathbf{H})} \geq 1$. Thus, $m \geq 1$. This establishes the global exponential stability of the periodic steady state.
\end{proof}

\section{Proof of Proposition \ref{gap_prop}}\label{appendix_c}

Recall that $\{J(q(t))=i\}$ denote the event that class $i$ is being served at time $t$. Note that $\pi_i(t) = P\{J(q(t))=i\}$. By the Law of Total Expectation, the average age $\bar{\Delta}_i(t)$ can be partitioned into events $\{J(q(t))=i\}$ and $\{J(q(t))\neq i\}$:
\begin{equation}
    \bar{\Delta}_i(t) = \mathbb{E}[\Delta_i(t) | J(q(t))=i] P\{J(q(t))=i\} + \mathbb{E}[\Delta_i(t) | J(q(t)) \neq i] P\{J(q(t)) \neq i\}. \notag
\end{equation}
Substituting $\Delta^s_i(t) = \mathbb{E}[\Delta_i(t) | J(q(t))=i]$ and $\Delta_i^c(t) = \mathbb{E}[\Delta_i(t) | J(q(t)) \neq i]$, we have:
\begin{equation}
    \bar{\Delta}_i(t) = \pi_i(t) \Delta^s_i(t) + (1 - \pi_i(t)) \Delta_i^c(t). \notag
\end{equation}
The gap between the service-conditioned average age $\Delta^s_i(t)$ and the average age $\bar{\Delta}_i(t)$ is then evaluated as:
\begin{align}
    \Delta^s_i(t) - \bar{\Delta}_i(t) &= \Delta^s_i(t) - \left[ \pi_i(t) \Delta^s_i(t) + (1 - \pi_i(t)) \Delta_i^c(t) \right] \notag \\
    &= (1 - \pi_i(t)) \Delta^s_i(t) - (1 - \pi_i(t)) \Delta_i^c(t) \notag \\
    &= (1 - \pi_i(t)) (\Delta^s_i(t) - \Delta_i^c(t)). \label{eq:gap_uncomp}
\end{align}
To connect the service-conditioned age $\Delta^s_i(t)$ to the average peak age $\hat{\Delta}_i(t)$,we break down the serving event by summing over the individual states in the state space $\mathcal{Q}$ in (\ref{eq:expand_over_Q}c).
\begin{subequations}
\refstepcounter{equation}
\begin{align} \label{eq:expand_over_Q}
\Delta^s_i(t) &\overset{\mathclap{(\alph{equation})}}{=} \mathbb{E}[\Delta_i(t) \mid J(q(t)) = i] \tag{\theparentequation} \\ \noalign{\refstepcounter{equation}}
&\overset{\mathclap{(\alph{equation})}}{=} \frac{\mathbb{E}[\Delta_i(t) 1{\{J(q(t)) = i\}}]}{P(J(q(t)) = i)} \nonumber \\ \noalign{\refstepcounter{equation}}
&\overset{\mathclap{(\alph{equation})}}{=} \frac{\mathbb{E}\left[\Delta_i(t) \sum_{s \in \mathcal{Q}} 1{\{J(s) = i\}} 1{\{q(t)=s\}}\right]}{\boldsymbol{\delta}_{J=i}^\top \mathbf{p}(t)} \nonumber \\ \noalign{\refstepcounter{equation}}
&\overset{\mathclap{(\alph{equation})}}{=} \frac{\sum_{s \in \mathcal{Q}} 1{\{J(s)=i\}} \mathbb{E}[\Delta_i(t) 1{\{q(t)=s\}}]}{\boldsymbol{\delta}_{J=i}^\top \mathbf{p}(t)} \notag \\ \noalign{\refstepcounter{equation}}
&\overset{\mathclap{(\alph{equation})}}{=} \frac{\sum_{s \in \mathcal{Q}} 1{\{J(s)=i\}} a_{i,s}(t)}{\boldsymbol{\delta}_{J=i}^\top \mathbf{p}(t)} \notag \\ \noalign{\refstepcounter{equation}}
&\overset{\mathclap{(\alph{equation})}}{=} \frac{\boldsymbol{\delta}_{J=i}^\top \mathbf{a}_i(t)}{\boldsymbol{\delta}_{J=i}^\top \mathbf{p}(t)} \nonumber \\  \noalign{\refstepcounter{equation}} 
&\overset{\mathclap{(\alph{equation})}}{=}\hat{\Delta}_i(t). \notag
\end{align}
\end{subequations}
Equation (\ref{eq:expand_over_Q}g) is proven in Appendix \ref{peak_AOI}. Because $\Delta^s_i(t) \equiv \hat{\Delta}_i(t)$, we can directly substitute the average peak age into \eqref{eq:gap_uncomp}. Hence, the gap between the average peak age and the average age can be written as:$$\hat{\Delta}_i(t) - \bar{\Delta}_i(t) = (1 - \pi_i(t)) (\hat{\Delta}_i(t) - \Delta_i^c(t)).$$

\section{Proof of Corollary \ref{fixed_point}}\label{appendix_d}

\begin{proof}
 Let $F$ denote the one-period map that takes an initial condition $\mathbf{x}(0)$ to the solution $\mathbf{x}(T)$ obtained by integrating the system of ODEs over one period $T$, such that $F(\mathbf{x}(0)) = \mathbf{x}(T)$. By definition, the periodic steady state $\mathbf{x}^\star(t)$ satisfies $\mathbf{x}^\star(0) = \mathbf{x}^\star(T) = F(\mathbf{x}^\star(0))$, making $\mathbf{x}^\star(0)$ a fixed point of $F$. To rigorously evaluate the convergence of the discrete steps while bypassing transient effects, we analyze the one-period map in the weighted norm $\|\cdot\|_\mathbf{H}$. By the construction of $\mathbf{H}$ established in the proof of Theorem \ref{existence_thm} in Appendix \ref{appendix_b}, the one-period map is free of the geometric multiplier $m$, yielding the strict contraction:
\begin{align}
\|\mathbf{x}(t) - \mathbf{x}^\star(t)\|_\mathbf{H} &= \|\mathbf{\Phi}(t) \big(\mathbf{x}(0) - \mathbf{x}^\star(0)\big)\|_\mathbf{H} \notag \\
&\leq \|\mathbf{F}(t)e^{\mathbf{K}t}(\mathbf{x}(0) - \mathbf{x}^\star(0))\|_\mathbf{H} \notag \\
&\leq \|\mathbf{F}(t)\|_\mathbf{H} \|e^{\mathbf{K}t}\|_\mathbf{H} \|\mathbf{x}(0) - \mathbf{x}^\star(0)\|_\mathbf{H}. \notag \\
&\leq \|\mathbf{F}(t)\|_\mathbf{H} e^{-\gamma t} \mathbf{x}(0) - \mathbf{x}^\star(0)\|_\mathbf{H}. \notag
\end{align}
Now, we evaluate the inequality at $t = T$.
\begin{align}
\|\mathbf{x}(T) - \mathbf{x}^\star(T)\|_\mathbf{H}&\leq \|\mathbf{F}(T)\|_\mathbf{H} e^{-\gamma T} \|\mathbf{x}(0) - \mathbf{x}^\star(0)\|_\mathbf{H}. \notag \\
&= \|\mathbf{F}(0)\|_\mathbf{H} e^{-\gamma T} \|\mathbf{x}(0) - \mathbf{x}^\star(0)\|_\mathbf{H}. \notag \\
&= \|\mathbf{I}\|_\mathbf{H} e^{-\gamma T} \|\mathbf{x}(0) - \mathbf{x}^\star(0)\|_\mathbf{H}. \notag \\
&= e^{-\gamma T} \|\mathbf{x}(0) - \mathbf{x}^\star(0)\|_\mathbf{H}. \notag 
\end{align}
 where $\|\mathbf{I}\|_\mathbf{H} = 1$ and $e^{-\gamma T} < 1$ since $\gamma > 0$ and $T > 0$. Therefore,
 \begin{align}
\|F(\mathbf{x}(0)) - F(\mathbf{x}^\star(0))\|_\mathbf{H} &\leq e^{-\gamma T} \|\mathbf{x}(0) - \mathbf{x}^\star(0)\|_\mathbf{H}. \notag 
\end{align}
 The relaxed fixed-point iteration updates the initial condition via:
 \begin{equation}\mathbf{x}^{(n+1)}(0) = (1 - \alpha)\mathbf{x}^{(n)}(0) + \alpha F(\mathbf{x}^{(n)}(0)). \notag
 \end{equation}
 Substituting this update rule and the fixed-point identity $x^\star(0) = (1 - \alpha)x^\star(0) + \alpha F(x^\star(0))$ into the error evaluated in the $H$-norm yields:
 \begin{align}\|\mathbf{x}^{(n+1)}(0) - \mathbf{x}^\star(0)\|_\mathbf{H} &= \|(1 - \alpha)(\mathbf{x}^{(n)}(0) - \mathbf{x}^\star(0)) + \alpha(F(\mathbf{x}^{(n)}(0)) - F(\mathbf{x}^\star(0)))\|_\mathbf{H} \notag \\
 &\leq (1 - \alpha) \|\mathbf{x}^{(n)}(0) - \mathbf{x}^\star(0)\|_\mathbf{H} + \alpha \|F(\mathbf{x}^{(n)}(0)) - F(\mathbf{x}^\star(0))\|_\mathbf{H} \notag \\
 &\leq (1 - \alpha) \|\mathbf{x}^{(n)}(0) - \mathbf{x}^\star(0)\|_\mathbf{H} + \alpha e^{-\gamma T} \|\mathbf{x}^{(n)}(0) - \mathbf{x}^\star(0)\|_\mathbf{H} \notag \\
 &= \left(1 - \alpha(1 - e^{-\gamma T})\right) \|\mathbf{x}^{(n)}(0) - \mathbf{x}^\star(0)\|_\mathbf{H}. \notag
 \end{align}
 Given $\alpha \in (0, 1]$ and $0 < e^{-\gamma T} < 1$, the iteration multiplier $\left(1 - \alpha(1 - e^{-\gamma T})\right)$ is strictly less than $1$. Consequently, the distance between the $n$-th iteration estimate $x^{(n)}(0)$ and the fixed point $x^\star(0)$ decreases monotonically in the $H$-norm, ensuring that $\|x^{(n)}(0) - x^\star(0)\|_H \to 0$ as $n \to \infty$. Because all vector norms in a finite-dimensional space are equivalent \cite[Theorem 2.4-5]{kreyszig1978introductory}, this guarantees asymptotic convergence to the fixed point in the standard Euclidean norm as well.

\end{proof}

\section{Fixed-Point Iteration for the Periodic Steady State}
\label{app:pss_algorithm}

This section records the relaxed fixed-point iteration used in Section~\ref{mumericalsolution}.

\begingroup
\setlength{\abovedisplayskip}{2pt}
\setlength{\belowdisplayskip}{2pt}
\begin{compactalgorithm}{Relaxed Fixed-Point Iteration for the Periodic Steady State}
\alginput{Period $T$; ODEs in Theorem~\ref{theorem_main}; tolerance $\varepsilon$; maximum iterations $K$; relaxation $\alpha\in(0,1]$.}
\algline{Choose $x^{(0)}(0)$, e.g., idle-state probability $1$ and all moments $0$.}
\algline{\textbf{for} $n=0,1,\ldots,K-1$ \textbf{do}}
\algline{\quad Integrate on $[0,T]$ and set $x^{(n)}(T)=F\bigl(x^{(n)}(0)\bigr)$.}
\algline{\quad Renormalize probabilities: $p_s \leftarrow p_s/\sum_{s'}p_{s'}$, for all $s$.}
\algline{\quad Update $x^{(n+1)}(0)\leftarrow(1-\alpha)x^{(n)}(0)+\alpha x^{(n)}(T)$.}
\algline{\quad \textbf{if} $\|x^{(n+1)}(0)-x^{(n)}(0)\|/(1+\|x^{(n)}(0)\|)\leq\varepsilon$ \textbf{then break}.}
\algline{\textbf{end for}}
\algoutput{$x^{\star}(0)\approx x^{(n)}(0)$.}
\end{compactalgorithm}
\endgroup

\begin{thebibliography}{10}
\providecommand{\url}[1]{#1}
\csname url@samestyle\endcsname
\providecommand{\newblock}{\relax}
\providecommand{\bibinfo}[2]{#2}
\providecommand{\BIBentrySTDinterwordspacing}{\spaceskip=0pt\relax}
\providecommand{\BIBentryALTinterwordstretchfactor}{4}
\providecommand{\BIBentryALTinterwordspacing}{\spaceskip=\fontdimen2\font plus
\BIBentryALTinterwordstretchfactor\fontdimen3\font minus
  \fontdimen4\font\relax}
\providecommand{\BIBforeignlanguage}[2]{{%
\expandafter\ifx\csname l@#1\endcsname\relax
\typeout{** WARNING: IEEEtran.bst: No hyphenation pattern has been}%
\typeout{** loaded for the language `#1'. Using the pattern for}%
\typeout{** the default language instead.}%
\else
\language=\csname l@#1\endcsname
\fi
#2}}
\providecommand{\BIBdecl}{\relax}
\BIBdecl

\bibitem{yates2021survey}
R.~D. Yates, Y.~Sun, D.~R. Brown, S.~Kaul, E.~Modiano, and S.~Ulukus, ``Age of
  information: An introduction and survey,'' \emph{{IEEE} J. Sel. Areas
  Commun.}, vol.~39, no.~5, pp. 1183--1210, 2021.

\bibitem{kadota2018scheduling}
I.~Kadota, A.~Sinha, E.~Uysal-Biyikoglu, R.~Singh, and E.~Modiano, ``Scheduling
  policies for minimizing age of information in broadcast wireless networks,''
  \emph{{IEEE/ACM} Trans. Netw.}, vol.~26, no.~6, pp. 2637--2650, 2018.

\bibitem{kaul2012real}
S.~Kaul, R.~Yates, and M.~Gruteser, ``Real-time status: How often should one
  update?'' in \emph{Proc. IEEE INFOCOM}, 2012, pp. 2731--2735.

\bibitem{yates2019multiple}
R.~D. Yates and S.~Kaul, ``The age of information: Real-time status updating by
  multiple sources,'' \emph{{IEEE} Trans. Inf. Theory}, vol.~65, no.~3, pp.
  1807--1827, 2019.

\bibitem{xu2021peakpriority}
J.~Xu and N.~Gautam, ``Peak age of information in priority queuing systems,''
  \emph{{IEEE} Trans. Inf. Theory}, vol.~67, no.~1, pp. 373--390, 2021.

\bibitem{costa2016packetmgmt}
M.~Costa, M.~Codreanu, and A.~Ephremides, ``On the age of information in status
  update systems with packet management,'' \emph{{IEEE} Trans. Inf. Theory},
  vol.~62, no.~4, pp. 1897--1910, 2016.

\bibitem{najm2017harq}
E.~Najm, R.~Yates, and E.~Soljanin, ``Status updates through {M/G/1/1} queues
  with {HARQ},'' in \emph{Proc. IEEE Int. Symp. Inf. Theory (ISIT)}, 2017, pp.
  131--135.

\bibitem{inoue2019stationarydist}
Y.~Inoue, H.~Masuyama, T.~Takine, and T.~Tanaka, ``A general formula for the
  stationary distribution of the age of information and its application to
  single-server queues,'' \emph{{IEEE} Trans. Inf. Theory}, vol.~65, no.~12,
  pp. 8305--8324, 2019.

\bibitem{champati2019distribution}
J.~P. Champati, H.~Al-Zubaidy, and J.~Gross, ``On the distribution of {AoI} for
  the {GI/GI/1/1} and {GI/GI/1/2} systems: Exact expressions and bounds,'' in
  \emph{Proc. IEEE INFOCOM}, 2019, pp. 37--45.

\bibitem{yates2020moments}
R.~D. Yates, ``The age of information in networks: Moments, distributions, and
  sampling,'' \emph{{IEEE} Trans. Inf. Theory}, vol.~66, no.~9, pp. 5712--5728,
  2020.

\bibitem{dogan2021probpreemptive}
O.~Dogan and N.~Akar, ``The multi-source probabilistically preemptive
  {M/PH/1/1} queue with packet errors,'' \emph{{IEEE} Trans. Commun.}, vol.~69,
  no.~11, pp. 7297--7308, 2021.

\bibitem{moltafet2022mgf}
M.~Moltafet, M.~Leinonen, and M.~Codreanu, ``Moment generating function of age
  of information in multi-source {M/G/1/1} queueing systems,'' \emph{{IEEE}
  Trans. Commun.}, vol.~70, no.~10, pp. 6503--6516, 2022.

\bibitem{hu2021violation}
L.~Hu, Z.~Chen, Y.~Dong, Y.~Jia, L.~Liang, and M.~Wang, ``Status update in
  {IoT} networks: Age-of-information violation probability and optimal update
  rate,'' \emph{{IEEE} Internet Things J.}, vol.~8, no.~14, pp.
  11\,329--11\,344, 2021.

\bibitem{xu2025distribution}
J.~Xu, W.~Wang, and N.~Gautam, ``On the distribution of age of information in
  time-varying updating systems,'' \emph{arXiv preprint arXiv:2507.03799},
  2025.

\bibitem{green2007coping}
L.~V. Green, P.~J. Kolesar, and W.~Whitt, ``Coping with time-varying demand
  when setting staffing requirements for a service system,'' \emph{Prod. Oper.
  Manag.}, vol.~16, no.~1, pp. 13--39, 2007.

\bibitem{feldman2008staffing}
Z.~Feldman, A.~Mandelbaum, W.~A. Massey, and W.~Whitt, ``Staffing of
  time-varying queues to achieve time-stable performance,'' \emph{Manage.
  Sci.}, vol.~54, no.~2, pp. 324--338, 2008.

\bibitem{antsaklis2006linear}
P.~J. Antsaklis and A.~N. Michel, \emph{Linear Systems}.\hskip 1em plus 0.5em
  minus 0.4em\relax Boston, MA: Birkh{\"a}user, 2006.

\bibitem{slane2011analysis}
J.~Slane and S.~Tragesser, ``Analysis of periodic nonautonomous inhomogeneous
  systems,'' \emph{Nonlinear Dynamics and Systems Theory}, vol.~11, no.~2, pp.
  183--198, 2011.

\bibitem{bernstein2009matrix}
D.~S. Bernstein, \emph{Matrix Mathematics: Theory, Facts, and Formulas},
  2nd~ed.\hskip 1em plus 0.5em minus 0.4em\relax Princeton, NJ: Princeton
  University Press, 2009.

\bibitem{hu2012bounds}
G.-D. Hu and T.~Mitsui, ``Bounds of the matrix eigenvalues and its exponential
  by lyapunov equation,'' \emph{Kybernetika}, vol.~48, no.~5, pp. 865--878,
  2012.

\bibitem{kreyszig1978introductory}
E.~Kreyszig, \emph{Introductory Functional Analysis with Applications}.\hskip
  1em plus 0.5em minus 0.4em\relax New York: John Wiley \& Sons, 1978.

\end{thebibliography}
\end{document}